%% file: main.tex
\documentclass{article}

\usepackage[square,numbers]{natbib}
\usepackage[final]{neurips_2022}

\usepackage[utf8]{inputenc} % allow utf-8 input
\usepackage[T1]{fontenc}    % use 8-bit T1 fonts
\usepackage{hyperref}       % hyperlinks
\usepackage{url}            % simple URL typesetting
\usepackage{booktabs}       % professional-quality tables
\usepackage{amsfonts}       % blackboard math symbols
\usepackage{nicefrac}       % compact symbols for 1/2, etc.
\usepackage{microtype}      % microtypography
\usepackage{xcolor}         % colors

\usepackage{subcaption}
\usepackage{amsmath,amssymb,amsthm}
\usepackage{causality}
\usepackage{algorithm}
\usepackage{algorithmicx}
\usepackage{algpseudocode}
\usepackage{lipsum}
\usepackage{wrapfig}

\newtheorem{theorem}{Theorem}

\newtheorem{lemma}{Lemma}
\theoremstyle{definition}
\newtheorem{definition}{Definition}
\newtheorem{example}{Example}

\newcommand{\G}{\mathcal{G}}

\algdef{SE}[SUBALG]{Indent}{EndIndent}{}{\algorithmicend\ }%
\algtext*{Indent}
\algtext*{EndIndent}

\title{Finding and Listing Front-door Adjustment Sets}

\author{%
  Hyunchai Jeong \\
  Purdue University \\
    \texttt{jeong3@purdue.edu} \\
   \And
   Jin Tian \\
   Iowa State University \\
   \texttt{jtian@iastate.edu} \\
   \And
   Elias Bareinboim \\
   Columbia University \\
   \texttt{eb@cs.columbia.edu} \\
}

\begin{document}

\maketitle

\begin{abstract}
Identifying the effects of new interventions from data is a significant challenge found across a wide range of the empirical sciences. A well-known strategy for identifying such effects is Pearl's \textit{front-door (FD) criterion} \citep{pearl:95a}.
The definition of the FD criterion is declarative, only allowing one to decide whether a specific set satisfies the criterion. In this paper, we present algorithms for finding and enumerating possible sets satisfying the FD criterion in a given causal diagram.
These results are useful in facilitating the practical applications of the FD criterion for causal effects estimation and helping scientists to select estimands with desired properties, e.g., based on cost, feasibility of measurement, or statistical power.
\end{abstract}

\input{1.intro}
\input{2.prelim}
\input{3.fd_find}
\input{4.fd_algorithm}
\input{5.fd_theorem}
\input{6.limitation}
\input{7.conclusion}

\bibliography{references}

\newpage
\appendix
\onecolumn

\input{8.appendix}

\end{document}

%% file: 1.intro.tex
\section{Introduction}

% - introduce causal inference
Learning cause and effect relationships is a fundamental challenge across data-driven fields.
For example, health scientists developing a treatment for curing lung cancer need to understand how a new drug affects the patient's body and the tumor's progression. 
The distillation of causal relations is indispensable to understanding the dynamics of the underlying system and how to perform decision-making in a principled and systematic fashion  \citep{pearl:2k,spirtes:etal93,bareinboim:pea16,PetersJanzingSchoelkopf17}.

% One of the most common methods for learning causal relations is through \textit{Randomized Controlled Trials} (RCTs) \citep{fisher:36}, considered as the ``gold standard'' in many fields of empirical research and are used throughout the health and social sciences as well as machine learning and AI.
One of the most common methods for learning causal relations is through \textit{Randomized Controlled Trials} (RCTs, for short) \citep{fisher:36}.
RCTs are considered as the ``gold standard'' in many fields of empirical research and are used throughout the health and social sciences as well as machine learning and AI.
In practice, however, RCTs are often hard to perform due to ethical, financial, and technical issues. For instance, it may be unethical to submit an individual to a certain condition if such condition may have some potentially negative effects (e.g., smoking). 
Whenever RCTs cannot be conducted, one needs to resort to analytical methods to infer causal relations from observational data, which appears in the literature as the problem of \textit{causal effect identification} \citep{pearl:95a,pearl:2k}. 
% in causal inference: identifying causal effects from observational data and causal assumptions represented as a causal diagram.

The causal identification problem asks whether the effect of holding a variable \(X\) at a constant value \(x\) on a variable \(Y\), written as \(P(Y | do(X = x))\), or \(P(Y | do(x))\), can be computed from a combination of observational data and causal assumptions.
% One of the most common ways of eliciting these assumptions is in the form of a causal diagram, which represents the underlying data generating process with nodes and edges.
One of the most common ways of eliciting these assumptions is in the form of a causal diagram represented by a directed acyclic graph (DAG), where its nodes and edges describe the underlying data generating process.
For instance, in Fig.~\ref{fig:fd:canonical}, three nodes \(X, Z, Y\) represent variables, a directed edge \(X \rightarrow Z\) indicates that \(X\) causes \(Z\), and 
a dashed-bidirected edge \(X \leftrightarrow Y\) represents that  \(X\) and \(Y\) are confounded by unmeasured (latent) factors. 
Different methods can solve the identification problem and a number of generalizations, including Pearl's celebrated do-calculus \citep{pearl:95a} as well as different algorithmic solutions  \citep{tian:pea02-general-id,shpitser:pea06a,huang:val06-identifiability,bareinboim:pea12-zid,lee:etal19,lee:20}.

In practice, researchers often rely on identification strategies that generate well-known identification formulas. One of the arguably most popular strategies is identification by covariate adjustment. 
Whenever a set $Z$ satisfies the \textit{back-door (BD) criterion} \citep{pearl:95a} relative to the pair $X$ and $Y$, where $X$ and $Y$ represent the treatment and outcome variables, respectively, the causal effect $P(Y | do(x))$ can be evaluated through the BD adjustment formula $\sum_z P(y|x,z) P(z)$. %The back-door criterion is only descriptive, i.e., it specifies whether a specific set $Z$ satisfies the back-door criterion or not, but does not provide a way to find an admissible back-door set $Z$. To facilitate the practical application of the back-door criterion, algorithms have been developed that allow one to find one or list all back-door admissible sets \citep{Takata2010,textor:11,zander:etal14,zander:etal19}.

Despite the popularity of the covariate adjustment technique for estimating causal effects,  there are still settings in which no BD admissible set exists.  For example, consider the causal diagram $\mathcal{G}$ in Fig.~\ref{fig:fd:canonical}. %One of the conditions of the back-door criterion is that the set $Z$ is not a descendant of the treatment $X$ (in the path leading to $Y$), which is clearly not the case for $Z$ in this particular diagram. 
There clearly exists no set to block the BD path from $X$ to $Y$, through the bidirected arrow,  $X \leftrightarrow Y$. 
One may surmise that this effect is not identifiable and the only one of evaluating the interventional distribution is through experimentation. 
Still, this is not the case. The effect $P(Y | do(x))$ is identifiable from $\mathcal{G}$ and the observed distribution $P(x,y,z)$ over $\{X,Y,Z\}$ by another classic identification strategy known as the \textit{front-door (FD) criterion} \citep{pearl:95a}. In particular, through the following FD adjustment formula provides the way of evaluating the interventional distribution:
% shown in Eq.~\ref{eq1}.
%Still, the set $Z$ satisfies what is known as the \textit{front-door criterion} \citep[Sec.~3.3.2]{pearl:2k}. The front-door criterion has a long history, and we defer to \citep{pearl:mackenzie2018}[Sec.~3.4] for a more detailed account; for applications of the front-door, see, e.g., \citep{glynn:kashin18}. In practice, front-door admissibility implies that the interventional distribution $P(Y | do(x))$ can be evaluated through the formula 
\begin{equation}
    \label{eq1}
    P(Y|do(x)) = \sum_{z} P(z|x) \sum_{x'} P(y|x',z) P(x').
\end{equation}
%when a set $Z$ is an admissible front-door adjustment set. 
We refer to \citet[Sec.~3.4]{pearl:mackenzie2018} for an interesting account of the history of the FD criterion, which was the first graphical generalization of the BD case. The FD criterion is drawing more attention in recent years. For applications of the FD criterion, see, e.g., \citet{hun:bar2019} and \citet{glynn:kashin18}. Statistically efficient and doubly robust estimators have recently been developed for estimating the FD estimand in Eq.~(\ref{eq1}) from finite samples %In practice, when one only has finite samples, the FD estimand in Eq.~(\ref{eq1}) enjoys statistically efficient and doubly robust estimators
\citep{fulcher2019robust}, which are still elusive for arbitrary estimands identifiable in a diagram despite recent progress \citep{jung2020estimating,jung2020werm,bhattacharya2020semiparametric,jung2021dmlid,xia:21}.  
%More recently, the front-door criterion has been extended to include a conditioning set $W$ that renders $Z$ admissible called \textit{conditional front-door criterion} \citep{hun:bar2019}. 

% Fig - Intro
\begin{wrapfigure}{R}{0.5\textwidth}
\begin{minipage}[t]{0.5\textwidth}
% \begin{figure}[t]
    \centering
    % filled in dummy text to improve aesthetics
    \null\hfill%
    \begin{subfigure}{0.37\textwidth}
        \includegraphics[width=\textwidth]{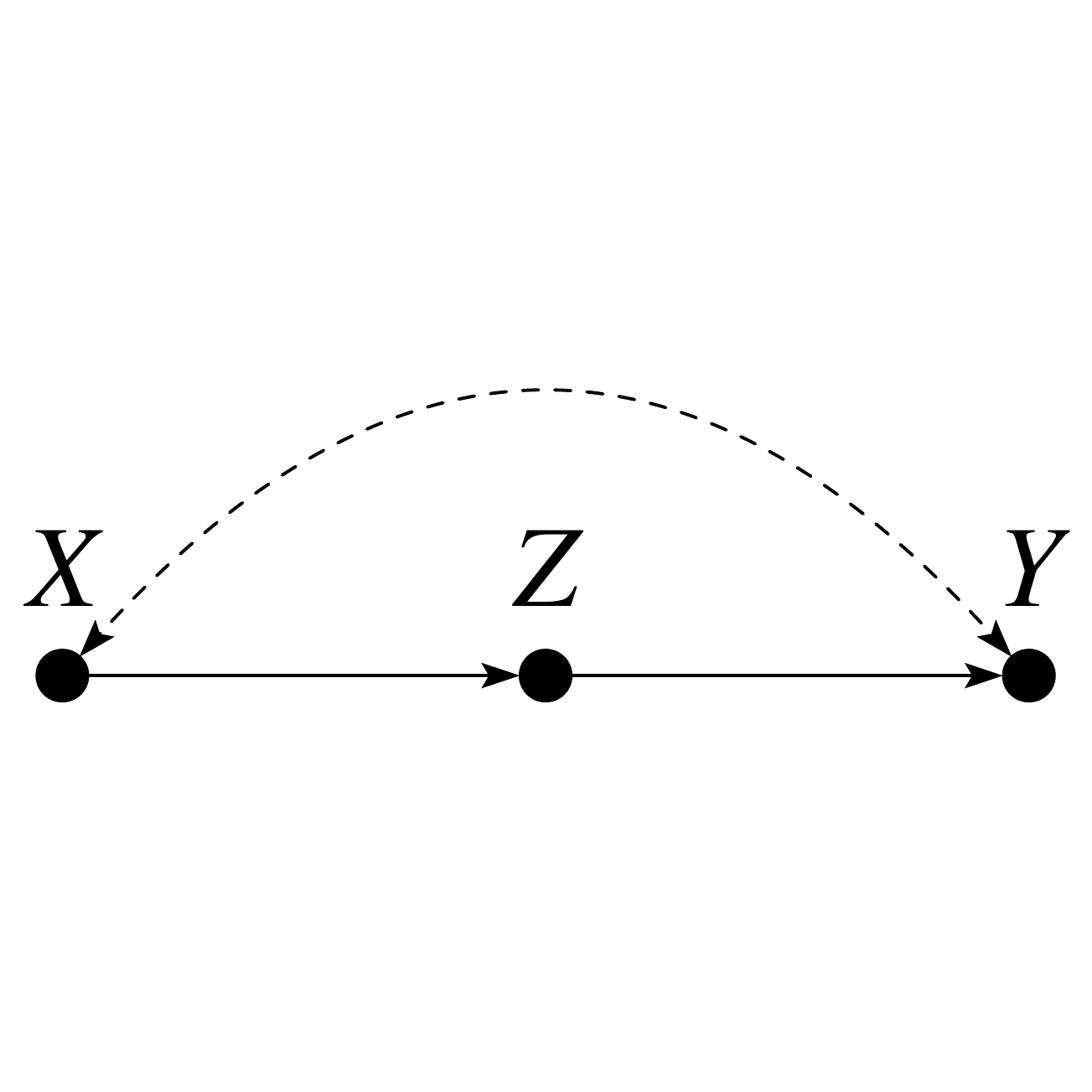}
        \caption{\(\G\)}
        \label{fig:fd:canonical}
    \end{subfigure}
    \hfill
    \begin{subfigure}{0.45\textwidth}
        \includegraphics[width=\textwidth]{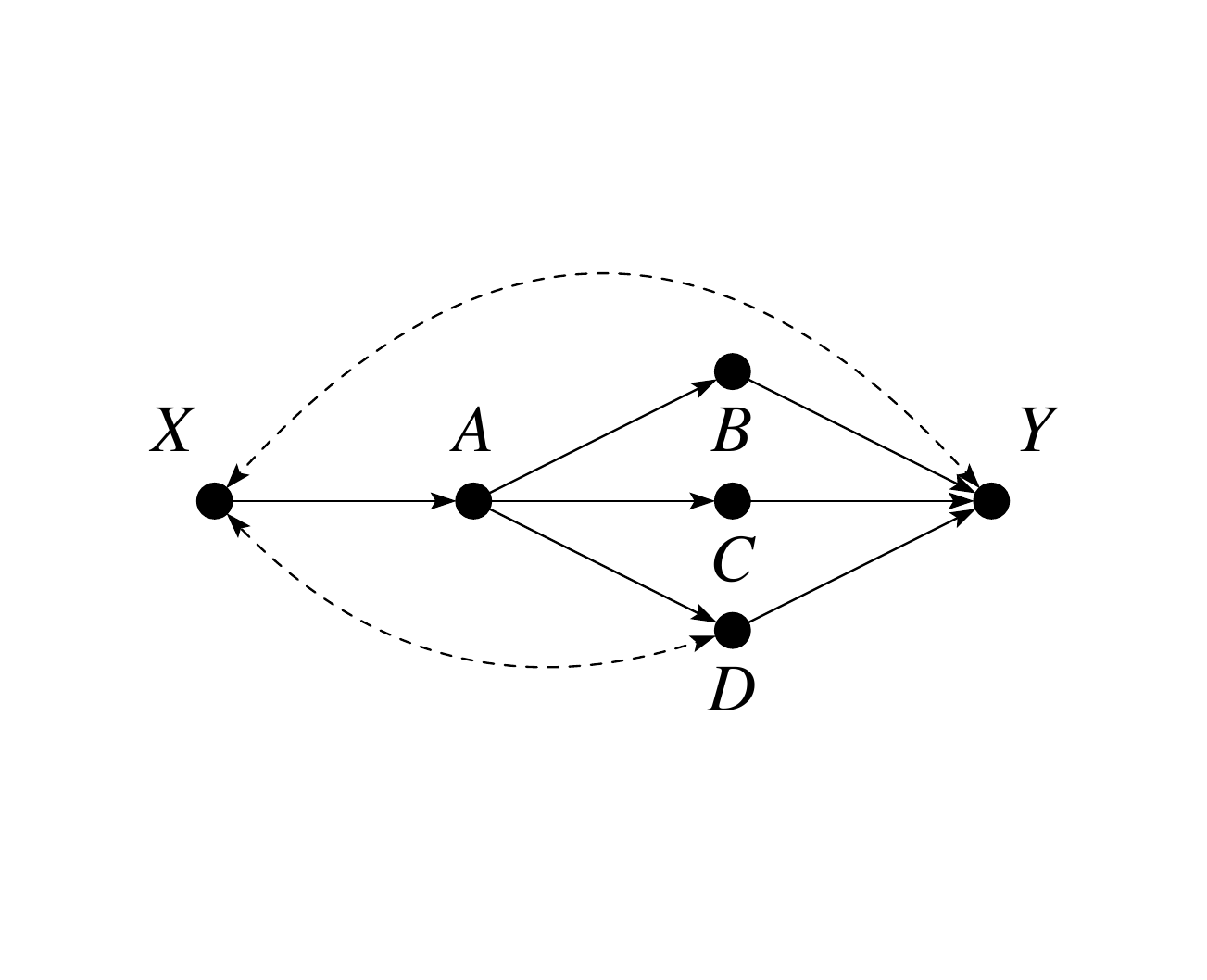}
        \caption{\(\G'\)}
        \label{fig:fd:intro}
    \end{subfigure}
    \null\hfill%
    
    % \begin{minipage}{0.20\textwidth}
    %     \centering
    %     \subfigure[\(\G_{1(a)}\)]{\label{fig:fd:canonical}
    %     \includegraphics[width=0.8\textwidth]{figures/fig_fd_canonical}}
    % \end{minipage}
    % \hfill
    % \begin{minipage}{0.25\textwidth}
    %     \centering
    %     \subfigure[\(\G_{1(b)}\)]{\label{fig:fd:intro}
    %     \includegraphics[width=0.8\textwidth]{figures/fig_fd_intro}}
    % \end{minipage}
\caption{
(a) A canonical example of the FD criterion where 
\(\{Z\}\) satisfies the FD criterion relative to \((\{X\},\{Y\})\).
In (b), four FD adjustment sets relative to \((\{X\},\{Y\})\) are available: \(\{A\}\), \(\{A,B\}\), \(\{A,C\}\), and \(\{A,B,C\}\).
}
% \end{figure}
\end{minipage}
\end{wrapfigure}

Both the BD and FD criteria are only descriptive, i.e., they specify whether a specific set $Z$ satisfies the criteria or not, but do not provide a way to find an admissible set $Z$. In addition, in many situations, it is possible that multiple adjustment sets exist. 
Consider for example the causal diagram in  Fig.~\ref{fig:fd:intro}, and the task of identifying the effect of $X$ on $Y$. The distribution $P(Y | do(x))$ can indeed be identified by the FD criterion with a set $Z=\{A, B, C\}$ given by the expression in Eq.~(\ref{eq1}) (with $Z$ replaced with $\{A, B, C\}$). 
Still, what if the variable $B$ is costly to measure or encodes some personal information about patients which is undesirable to be shared due to ethical concerns? In this case, the set $Z=\{A,C\}$ also satisfies the FD criterion and may be used. 
Even when both \(B\) and \(C\) are unmeasured, the set \(Z=\{A\}\) is also FD admissible.

%In Fig.~\ref{fig:fd:canonical}, $Z$ is the only admissible front-door adjustment set. 
%In reality, 
%In many situations, however, it is possible that multiple front-door adjustment sets exist. For instance, in the causal diagram in  Fig.~\ref{fig:fd:intro},  the effect $P(Y | do(x))$ can be identified by the front-door criterion with a set $Z=\{A, B, C\}$ given by the expression in Eq.~(\ref{eq1}) (with $Z$ replaced with $\{A, B, C\}$). Still, what if the variable $B$ is costly to measure or encodes some personal information about patients which is undesirable to be shared due to ethical concerns? In this case, the set $Z=\{A,C\}$ also satisfies the front-door criterion and may be used. % and the effect can be evaluated through the expression 
%\begin{equation}
%    \label{eq:fd:ac}
%    P(y|do(x)) = \sum_{a,c} P(a,c|x) \sum_{x'} P(y|x',a,c) P(x')
%\end{equation}
%Even when both \(B\) and \(C\) are unmeasured, the set \(Z=\{A\}\) is front-door admissible. 
%In fact, in this particular case, \(\{A\}\) is minimal in a sense that no variable can be removed from the set while still retaining admissibility. Minimal sets may be favorable by domain experts as those could alleviate certain challenges, such as the cost for measurement and the technical feasibility of measuring. 

This simple example shows that a target effect can be estimated using different adjustment sets leading to different probability expressions %from multiple probability distributions, each 
over different set of variables, which has important practical implications. 
Each variable implies different practical challenges in terms of measurement, such as cost, availability, privacy.
Each estimand has different statistical properties in terms of sample complexity, variance, which may play a key role in the study design \citep{robins:92,hahn:98,rotnitzky:20,smucler:22}.
%If given a list of candidate front-door admissible sets,  a scientist may select  a set that exhibits desirable properties.
Algorithms for finding and listing all possible adjustment sets are hence very useful in practice, which will allow scientists to select an adjustment set that exhibits desirable properties. Indeed, algorithms have been developed in recent years for finding one or listing all BD admissible sets \citep{Takata2010,textor:11,zander:etal14,perkovic:etal18,zander:etal19}. However, no such algorithm is currently available for finding/listing FD admissible sets.

%While the front-door criterion allows one to decide whether a specific set $Z$ is admissible with respect to a pair $(X, Y)$ in a given causal diagram, %there is no systematic way of navigating over the space of all admissible sets so as to then possibly select sets that have certain desirable properties. 
%no systematic way is currently available for finding an admissible set or listing all possible admissible sets.  
The goal of this paper is to close this gap to facilitate the practical applications of the FD criterion for causal effects estimation and help scientists to select estimand with certain desired properties \footnote{Code is available at \url{https://github.com/CausalAILab/FrontdoorAdjustmentSets}.}.
Specifically, the  contributions of this paper are as follows:

\begin{enumerate}
    \item We develop an algorithm that finds an admissible front-door adjustment set \(\*Z\) in a given causal diagram in polynomial time (if one exists). We solve a variant of the problem that imposes constraints $\*I \subseteq \*Z \subseteq \*R$ for given sets $\*I$ and $\*R$, %(to be specified later), 
    which allows a scientist to constrain the search to include specific subsets of variables or exclude variables from search perhaps due to cost, availability, or other technical considerations.

    \item We develop a sound and complete algorithm that enumerates all front-door adjustment sets with polynomial delay - the algorithm takes polynomial amount of time to return each new admissible set, if one exists, or return failure whenever it exhausted all admissible sets.
    
    % \item We study the problem of finding minimal front-door adjustment sets.
    % A set is said to be minimal if none of its subset is also admissible. In this case, domain experts can select one admissible set from the list that fits their best interest (e.g., low measurement cost, lower asymptotic variance).
    
    % \item We study the problem of finding a minimal front-door adjustment set.
    % A set is said to be minimal if none of its subset is also admissible.
    % We present a function that finds a minimal front-door adjustment set \(\*Z\) in polynomial time, under a constraint that \(\*Z\) must include a set of variables \(\*I\).
    
   % \item We extend these results to develop an algorithm that finds a conditional front-door adjustment set \(\*Z\) along with a conditioning set $\*W$ 
    %\textit{front-door adjustment pair} $(\*Z, \*W)$ 
    %in a given causal diagram in polynomial time and an algorithm that enumerates all conditional front-door adjustment sets with polynomial delay.
    %We study the problem of finding conditional front-door adjustment pairs. We present a function that finds a conditional front-door adjustment pairs \((\*Z,\*W)\) in polynomial time, under a constraint that \(\*Z\) must include a set of variables \(\*I\). Also, we develop a sound and complete algorithm that enumerates all conditional front-door adjustment pairs with polynomial delay.
\end{enumerate}

%% file: 2.prelim.tex
\section{Preliminaries}

% variables, family relations
\noindent \textbf{Notation.} 
We write a variable in capital letters (\(X\)) and its value as small letters (\(x\)). 
Bold letters, \(\*X\) or \(\*x\), represent a set of variables or values.
We use kinship terminology to denote various relationships
in a graph $\G$ and denote the parents, ancestors, and descendants of \(\*X\) (including \(\*X\) itself) as \(\Pa{\*X}, \An{\*X}\), and \(\De{\*X}\), respectively. 
Given a graph $\G$ over a set of variables $\*V$, a subgraph \(\G_{\*X}\) consists of a subset of variables \(\*X \subseteq \*V\) and their incident edges in \(\G\).
A graph \(\G\) can be transformed: \(\G_{\overline{\*X}}\) is the graph resulting from removing all incoming edges to \(\*X\), and \(\G_{\underline{\*X}}\) is the graph with all outgoing edges from \(\*X\) removed.
A DAG \(\G\) may be \textit{moralized} into an undirected graph where all directed edges of \(\G\) are converted into undirected edges, and for every pair of nonadjacent nodes in \(\G\) that share a common child, an undirected edge that connects such pair is added \citep{lauritzen:spi88}.

A path \(\pi\) from a node \(X\) to a node \(Y\) in \(\G\) is a sequence of edges where \(X\) and \(Y\) are the endpoints of \(\pi\).
A node \(W\) on \(\pi\) is said to be a collider if \(W\) has converging arrows into \(W\) in \(\pi\), e.g., \(\rightarrow W \leftarrow\) or \(\leftrightarrow W \leftarrow\).
\(\pi\) is said to be blocked by a set \(\*Z\) if there exists a node \(W\) on \(\pi\) satisfying one of the following two conditions: 1) \(W\) is a collider, and neither \(W\) nor any of its descendants are in \(\*Z\), or 2) \(W\) is not a collider, and \(W\) is in \(\*Z\) \citep{pearl:88a}.
Given three disjoint sets \(\*X,\*Y\), and \(\*Z\) in \(\G\), \(\*Z\) is said to $d$-separate \(\*X\) from \(\*Y\) in \(\G\) if and only if \(\*Z\) blocks every path from a node in \(\*X\) to a node in \(\*Y\) according to the $d$-separation criterion \citep{pearl:88a}, and we say that \(\*Z\) is a \emph{separator} of \(\*X\) and \(\*Y\) in \(\G\).

% % Fig - Intro
% \begin{figure}[t]
%     \centering
    
%     \begin{subfigure}{0.2\textwidth}
%         \includegraphics[width=\textwidth]{figures/fig_fd_canonical}
%         \caption{\(\G_{1a}\)}
%         \label{fig:fd:canonical}
%     \end{subfigure}
%     \hfill
%     \begin{subfigure}{0.2\textwidth}
%         \includegraphics[width=\textwidth]{figures/fig_cfd_intro}
%         \caption{\(\G_{1b}\)}
%         \label{fig:cfd:intro}
%     \end{subfigure}

% \caption{
% A canonical example of the front-door criterion is shown in (a).
% \(\{Z\}\) satisfies the front-door criterion relative to \((\{X\},\{Y\})\) and is minimal.
% In (b), five conditional front-door adjustment pairs relative to \((\{X\},\{Y\})\) are available: \((\{A\},\{D\})\), \((\{A,B\},\{D\})\), \((\{A,C\},\{D\})\), \((\{B,C\},\{D\})\), and \((\{A,B,C\},\{D\})\).
% However, only two of those are minimal: \((\{A\},\{D\})\) and \((\{B,C\},\{D\})\).
% }
% \end{figure}

% SCM
\noindent \textbf{Structural Causal Models (SCMs).} 
We use Structural Causal Models (SCMs, for short) \citep{pearl:2k} as our basic semantical framework.
An SCM is a 4-tuple \(\langle \*U, \*V, \*F, P(\*u) \rangle\), where 1) \(\*U\) is a set of exogenous (latent) variables, 2) \(\*V\) is a set of endogenous (observed) variables, 3) \(\*F\) is a set of functions \(\{f_V\}_{V \in \*V}\) that determine the value of endogenous variables, e.g., \(v \gets f_V(\*{pa}_V, \*u_V)\) is a function with \(\*{PA}_V \subseteq \*V \setminus \{V\}\) and \(\*U_V \subseteq \*U\), and 4) \(P(\*u)\) is a joint distribution over the exogenous variables \(\*U\).
Each SCM induces a \emph{causal diagram} \(\G\) \citep[Def.~13]{bareinboim:etal20} where every variable \(v \in \*V\) is a vertex and directed edges in \(\G\) correspond to functional relationships as specified in \(\*F\) and dashed bidirected edges represent common exogenous variables between two vertices.
Within the structural semantics, performing an intervention and setting $X=x$ is represented through the do-operator, $do(X=x)$, which encodes the operation of replacing the original functions of $X$ (i.e., $f_X(\*{pa}_X,\*u_X)$) by the constant $x$ and induces a submodel $\mathcal{M}_{x}$ and an interventional distribution $P(v \vert do(x))$. 

% Yes, the paper is assuming causal ADMGs or causal diagrams introduced in Section 2 - Preliminaries under the paragraph Structural Causal Models (SCMs). We will state this more clearly.

\noindent \textbf{Classic Causal Effects Identification Criteria.} 
Given a causal diagram $\G$ over $\*{V}$, an effect \(P(\*y|do(\*x))\) is said to be 
\emph{identifiable} in $\G$ if \(P(\*y|do(\*x))\) is uniquely computable from the observed distribution $P(\*{v})$ in any SCM  that induces $\G$ \citep[p.~77]{pearl:2k}.

A path between $X$ and $Y$ with an arrow into $X$ is known as a \emph{back-door} path from $X$ to $Y$. The celebrated back-door (BD) criterion \citep{pearl:95a} provides a sufficient condition for effect identification from observational data, which states that if a set \(\*Z\) of non-descendants of \(\*X\) blocks all BD paths from \(\*X\) to \(\*Y\), then the causal effect \(P(\*y|do(\*x))\) is identified by the BD adjustment formula:

\begin{equation}\label{eq:bd}
   P(\*y|do(\*x)) = \sum_{\*z} P(\*y|\*x, \*z) P(\*z)
\end{equation}

%The celebrated back-door (BD) criterion \citep{pearl:95a} provides a sufficient condition for effect identification from observational data. Please refer to the paper for the details.
Another classic identification condition that is key to the discussion in this paper is known as the front-door criterion, which is defined as follows: 

\begin{definition}{(Front-door (FD) Criterion \citep{pearl:95a})}
A set of variables \(\*Z\) is said to satisfy the front-door criterion relative to the pair \((\*X, \*Y)\) if
\begin{enumerate}
    \item \(\*Z\) intercepts all directed paths from \(\*X\) to \(\*Y\),
    \item There is no unblocked back-door path from \(\*X\) to \(\*Z\), and
    \item All back-door paths from \(\*Z\) to \(\*Y\) are blocked by \(\*X\), i.e.,  \(\*X\) is a separator of \(\*Z\) and \(\*Y\)  in \(\G_{\underline{\*Z}}\).
\end{enumerate}
\end{definition}

If \(\*Z\) satisfies the FD criterion relative to the pair \((\*X, \*Y)\), then \(P(\*y|do(\*x))\) is identified by the following FD adjustment formula \citep{pearl:95a}:

\begin{equation}\label{eq:fd}
    P(\*y|do(\*x)) = \sum_{\*z} P(\*z|\*x) \sum_{\*x'} P(\*y|\*x',\*z) P(\*x').
\end{equation}

%% file: 3.fd_find.tex
\section{Finding A Front-door Adjustment Set}
\label{section:fd_find}

%We present an efficient method to address the following question: given a causal graph \(\G\), find a set of covariates \(\*Z\) that allows us to identify \(P(\*y | do(\*x))\) by FD adjustment.  We may ask variants of the question, such as adding a constraint \(\*I \subseteq \*Z \subseteq \*R\) for given sets \(\*I\) and \(\*R\). We use the function as a primary building block to develop an algorithm that enumerates all FD adjustment sets in the next section.
% Function - FindFDSet
\begin{wrapfigure}{R}{0.53\textwidth}
\begin{minipage}[t]{0.53\textwidth}
\begin{algorithm}[H]
\caption{\textsc{FindFDSet} (\(\G, \*X, \*Y, \*I, \*R\))}
\label{alg:findfdset}

\begin{algorithmic} [1]
    \State {\bfseries Input:} \(\G\) a causal diagram; \(\*X,\*Y\) disjoint sets of variables; \(\*I,\*R\) sets of variables.
    
    \State {\bfseries Output:} \(\*Z\) a set of variables satisfying the front-door criterion relative to \((\*X,\*Y)\) with the constraint \(\*I \subseteq \*Z \subseteq \*R\).
    
    % \hspace*{\algorithmicindent} \textbf{Output}: Listing minimal front-door adjustment set \(\*Z\) relative to \((\*X,\*Y)\).
    
    % \State \textbf{if} \(\*I \cap (\An{\*X} \cup \De{\*Y}) \neq \emptyset\) \textbf{then:}
    % \textbf{return} \(\perp\)
    
    % \State \(\*R \gets \*R \setminus (\An{\*X} \cup \De{\*Y})\)
    
    % \State \textbf{if} \(\*I \cap \De{\*Y} \neq \emptyset\) \textbf{then:}
    % \textbf{return} \(\perp\)
        % \Indent
        % \State \textbf{return} \(\perp\)
        % \EndIndent
    
    % \State \(\*R \gets \*R \setminus \De{\*Y}\)
    % \State \textbf{if} \(\*R = \emptyset\) \textbf{then:}
    % \textbf{return} \(\perp\)
    
    \State \textbf{Step 1:}
        \Indent
        % \begin{ALC@g}
        % \bindent
        \State \(\*R' \gets \textsc{GetCand2ndFDC}(\G, \*X, \*I, \*R)\)
        % \eindent
        % \end{ALC@g}
        
        % \If {\(\*R' = \emptyset\)}
        %     \State \textbf{return} \(\perp\)
        % \EndIf
        
        \State \textbf{if} \(\*R' = \perp\) \textbf{then:}
        \textbf{return} \(\perp\)
            % \Indent
            % \State \textbf{return} \(\perp\)
            % \EndIndent
        
        \EndIndent
    
    \State \textbf{Step 2:}
        \Indent
        \State \(\*R'' \gets \textsc{GetCand3rdFDC}(\G, \*X, \*Y, \*I, \*R')\)
        
        % \If {\(\*R'' = \emptyset\)}
        %     \State \textbf{return} \(\perp\)
        % \EndIf
        
        \State \textbf{if} \(\*R'' = \perp\) \textbf{then:}
        \textbf{return} \(\perp\)
            % \Indent
            % \State \textbf{return} \(\perp\)
            % \EndIndent
        
        \EndIndent
    
    \State \textbf{Step 3:}
        \Indent
        % \State \(\*Z \gets \*R''\)
        \State \(\G' \gets \textsc{GetCausalPathGraph}(\G, \*X, \*Y)\)
        
        % \If {\(\textsc{TestSep}(\G', \*X, \*Y, \*Z) = \) True}
        %     \State \textbf{return} \(\*Z\)
        % \Else
        %     \State \textbf{return} \(\perp\)
        % \EndIf
        
        \State \textbf{if} \(\textsc{TestSep}(\G', \*X, \*Y, \*R'') = \) True \textbf{then:} \label{findfdset:testsep}
        % \textbf{return} \(\*Z\)
            \Indent
            \State \textbf{return} \(\*Z =\*R''\)
            \EndIndent
        \State \textbf{else:}
        \textbf{return} \(\perp\)
            % \Indent
            % \State \textbf{return} \(\perp\)
            % \EndIndent
        
        \EndIndent
\end{algorithmic}
\end{algorithm}
\end{minipage}
\end{wrapfigure}
In this section, we address the following question: given a causal diagram \(\G\), is there  a set  \(\*Z\) that satisfies the FD criterion relative to the pair \((\*X, \*Y)\) and,  therefore, allows us to identify \(P(\*y | do(\*x))\) by the FD adjustment? 
We %start by solving 
solve a more general variant of this question that imposes a constraint \(\*I \subseteq \*Z \subseteq \*R\) for given sets \(\*I\) and \(\*R\). %Here, \(\*I\) means that the set should be included and \(\*R\) means the variables that could be included.
Here, \(\*I\) are variables that must be included in \(\*Z\) (\(\*I\) could be empty) and \(\*R\) are variables that could be included in \(\*Z\) (\(\*R\) could be $\*V \setminus (\*X \cup \*Y)$). Note the constraint that variables in \(\*W\) \textit{cannot} be included can be enforced by excluding \(\*W\) from \(\*R\).
%Also, we can enforce the constraint that a certain set of variables \(\*W\) \textit{cannot} be included in any FD adjustment set by excluding \(\*W\) from \(\*R\).
Solving this version of the problem will allow scientists to put constraints on candidate adjustment sets based on practical considerations. In addition, this version will form a building block for an algorithm that enumerates all FD admissible sets in a given $\G$ - the algorithm \textsc{ListFDSets} (shown in Alg.~\ref{alg:listfdsets} in Section~\ref{section:fd_algorithm}) for listing all FD admissible sets will utilize this result during the recursive call. % and navigating the sets through \(\*I\)'s and \(\*Z\)'s.

%There are two reasons for solving the problem with the constraint \(\*I \subseteq \*Z \subseteq \*R\). One is more practical; a scientist can put constraints on candidate adjustment sets, e.g., excluding variables from consideration by excluding them from \(\*R\). Another is a bit more technical and relates to the next step of listing all admissible sets. The algorithm \textsc{ListFDSets} (shown in Alg.~\ref{alg:listfdsets}) for listing all adjustment sets needs solving this variant during the recursive call and navigating the sets through \(\*I\)'s and \(\*Z\)'s. The details of \textsc{ListFDSets} will be explained in Section~\ref{section:fd_algorithm}.
We have developed a procedure called \textsc{FindFDSet} shown in Alg.~\ref{alg:findfdset} that outputs a FD adjustment set \(\*Z\) relative to \((\*X, \*Y)\) satisfying \(\*I \subseteq \*Z \subseteq \*R\), or outputs $\perp$ if none exists, given a causal diagram \(\G\), disjoint sets of variables \(\*X\) and \(\*Y\), and two sets of variables \(\*I\) and \(\*R\).

\begin{example} \label{ex:findfd}
% \label{ex:listfdsets:sample}
Consider the causal graph \(\G'\), shown in Fig.~\ref{fig:fd:intro}, with \(\*X = \{X\}\), \(\*Y = \{Y\}\), \(\*I = \emptyset\) and \(\*R = \{A,B,C,D\}\).
Then, \textsc{FindFDSet} outputs \(\{A,B,C\}\).
With \(\*I = \{C\}\) and \(\*R = \{A,C\}\), \textsc{FindFDSet} outputs \(\{A,C\}\).
With \(\*I = \{D\}\) and \(\*R = \{A,B,C,D\}\), \textsc{FindFDSet} outputs \(\perp\) as no FD adjustment set that contains $D$ is available.
\end{example}

\begin{wrapfigure}{R}{0.53\textwidth}
% \begin{minipage}[t]{0.55\textwidth}
% \begin{figure}

\begin{algorithmic}[1]
\Function {GetCand2ndFDC} {$\G, \*X, \*I, \*R$}
% \Function {\textsc{GetCand2ndFDC}($\G, \*X, \*Y, \*I, \*R$)}

    % \State {\bfseries Input:} \(\G\) a causal diagram; \(\*X,\*Y\) disjoint sets of variables; \(\*I,\*R\) sets of variables.
    
    \State {\bfseries Output:} $\*R'$ with \(\*I \subseteq \*R' \subseteq \*R\), the set of candidate variables consisting of all the variables \(v \in \*R\) such that there is no BD path from \(\*X\) to \(v\).

    \State \(\*R' \gets \*R\)
    
    % \ForAll {\(v \in \*R\)}
    %     \If {\(\textsc{TestSep}(\G_{\underline{\*X}}, \*X, v, \emptyset) = \) False} \label{func:getcand2ndfdc:test}
    %         \If {\(v \in \*I\)}
    %             \State \textbf{return} \(\emptyset\)
    %         \Else
    %             \State \(\*R' \gets \*R' \setminus \{v\}\)
    %         \EndIf
    %     \EndIf
    % \EndFor \label{func:getcand2ndfdc:end}
    
    \State \textbf{for all} \(v \in \*R\):
        \Indent
        \State \textbf{if} \(\textsc{TestSep}(\G_{\underline{\*X}}, \*X, v, \emptyset) = \) False \textbf{then:} \label{func:getcand2ndfdc:test}
            \Indent
            \State \textbf{if} \(v \in \*I\) \textbf{then:} \textbf{return} \(\perp\)
                % \Indent
                %     \State \textbf{return} \(\emptyset\)
                % \EndIndent
            \State \textbf{else:} \(\*R' \gets \*R' \setminus \{v\}\)
                % \Indent
                %     \State \(\*R' \gets \*R' \setminus \{v\}\)
                % \EndIndent
            \EndIndent
        \EndIndent
    \State \textbf{end for} \label{func:getcand2ndfdc:end}
    
    \State \textbf{return} \(\*R'\)
\EndFunction
\end{algorithmic}
\caption{A function that outputs the set of candidate variables satisfying the second condition of the FD criterion.}
\label{func:getcand2ndfdc}
% \end{minipage}
% \end{figure}
\end{wrapfigure}
\textsc{FindFDSet} runs in three major steps.
Each step identifies candidate variables that incrementally satisfy each of the conditions of the FD criterion relative to \((\*X,\*Y)\).
First, \textsc{FindFDSet} constructs a set of candidate variables \(\*R'\), with \(\*I \subseteq \*R' \subseteq \*R\), such that every subset \(\*Z\) with \(\*I \subseteq \*Z \subseteq \*R'\) satisfies the second condition of the FD criterion (i.e., there is no BD path from \(\*X\) to \(\*Z\)).
Next, \textsc{FindFDSet} generates a set of candidate variables \(\*R''\), with \(\*I \subseteq \*R'' \subseteq \*R'\), such that for every variable \(v \in \*R''\), there exists a set \(\*Z\) with \(\*I \subseteq \*Z \subseteq \*R'\) and \(v \in \*Z\) that further satisfies the third condition of the FD criterion, that is, all BD paths from \(\*Z\) to \(\*Y\) are blocked by \(\*X\).
Finally, \textsc{FindFDSet} outputs a set \(\*Z\) that further satisfies the first condition of the FD criterion - \(\*Z\) intercepts all causal paths from \(\*X\) to \(\*Y\).
% We explain the details of each step one by one in the following.

\subsection*{Step 1 of \textsc{FindFDSet}}

In Step 1, \textsc{FindFDSet} calls the function \textsc{GetCand2ndFDC} (presented in Fig.~\ref{func:getcand2ndfdc}) to construct a set \(\*R'\) that consists of all the variables \(v \in \*R\) such that there is no BD path from \(\*X\) to \(v\) (\(\*R'\) is set to empty if there is a BD path from \(\*X\) to \(\*I\)).
Then, there is no BD path from \(\*X\) to any set \(\*I \subseteq\*Z \subseteq \*R'\) since, by definition, there is no BD path from \(\*X\) to \(\*Z\) if and only if there is no BD path from \(\*X\) to any \(v \in \*Z\).
%By definition, there is no BD path from \(\*X\) to any set \(\*W\) if and only if there is no BD path from every \(x \in \*X\) to every \(u \in \*W\). Then, there is no BD path from \(\*X\) to any set \(\*Z\) with \(\*I \subseteq \*Z \subseteq \*R'\) since for every \(v \in \*R'\), there is no BD path from \(\*X\) to \(v\). \(\*Z\) satisfies the second condition of the FD criterion relative to \((\*X,\*Y)\).

% \begin{lemma}
% \label{lemma:findfdset:backdoor}
% Let \(\*X,\*Z\) be disjoint sets of variables.
% Then, there is no back-door path from \(\*X\) to \(\*Z\) if and only if there is no back-door path from \(\*X\) to every \(v \in \*Z\).
% \end{lemma}

% \textsc{GetCand2ndFDC} iterates through each variable \(v \in \*R\) and checks if there exists an open BD path from \(\*X\) to \(v\) by calling the function  \(\textsc{TestSep}(\G_{\underline{\*X}}, \*X, v, \emptyset)\) given by \cite{zander:etal14}.
% \textsc{GetCand2ndFDC} iterates through each variable \(v \in \*R\) and checks if there exists an open BD path from \(\*X\) to \(v\).
% This can be achieved by calling the function \(\textsc{TestSep}(\G_{\underline{\*X}}, \*X, v, \emptyset)\) given by \cite{zander:etal14}.
\textsc{GetCand2ndFDC} iterates through each variable \(v \in \*R\) and checks if there exists an open BD path from \(\*X\) to \(v\) by calling the function  \(\textsc{TestSep}(\G_{\underline{\*X}}, \*X, v, \emptyset)\) \citep{zander:etal14}.
\(\textsc{TestSep}(\G, \*A, \*B, \*C)\) returns True if $\*C$ is a separator of $\*A$ and $\*B$ in \(\G\), or False otherwise.
Therefore, \(\textsc{TestSep}(\G_{\underline{\*X}}, \*X, v, \emptyset)\) returns True if \(\emptyset\) is a separator of \(\*X\) and \(v\) in \(\G_{\underline{\*X}}\) (i.e., there is no BD path from \(\*X\) to \(v\)), or False otherwise.
If \textsc{TestSep} returns False, then \(v\) is removed from \(\*R'\) because every set \(\*Z \) containing \(v\) violates the second condition of the FD criterion relative to \((\*X,\*Y)\).

\begin{example}\label{eg-2ndFDC}
\label{ex:findfdset:sample}
Continuing Example~\ref{ex:findfd}. With \(\*I = \emptyset\) and \(\*R = \{A,B,C,D\}\), \textsc{GetCand2ndFDC} outputs a set \(\*R' = \{A,B,C\}\).
\(D\) is excluded from \(\*R'\) since there exists a BD path from \(\{X\}\) to \(\{D\}\), and any set containing \(D\) violates the second condition of the FD criterion relative to \((\{X\},\{Y\})\).
\end{example}

\begin{lemma}[Correctness of \textsc{GetCand2ndFDC}]
\label{lemma:getcand2ndfdc}
%Let \(\G\) be a causal diagram, \(\*X,\*Y\) disjoint sets of variables, and \(\*I, \*R\) sets of variables where \(\*I \subseteq \*R\).
\textsc{GetCand2ndFDC}\((\G, \*X, \*I, \*R)\) generates a set of variables \(\*R'\) with \(\*I \subseteq \*R' \subseteq \*R\) such that \(\*R'\) consists of all and only variables \(v\) that satisfies the second condition of the FD criterion relative to \((\*X,\*Y)\).
Further, every subset \(\*Z \subseteq \*R'\) satisfies the second condition of the FD criterion relative to \((\*X,\*Y)\), and every set \(\*Z\) with \(\*I \subseteq \*Z \subseteq \*R\) that satisfies the second condition of the FD  criterion relative to \((\*X,\*Y)\) must be a subset of \(\*R'\).
\end{lemma}

\subsection*{Step 2 of \textsc{FindFDSet}}

% Function - GetCand3rdFDC
\begin{wrapfigure}{R}{0.54\textwidth}
% \begin{minipage}[t]{0.53\textwidth}
% \begin{figure}

\begin{algorithmic}[1]
\Function {GetCand3rdFDC} {$\G, \*X, \*Y, \*I, \*R'$}
% \Function {\textsc{GetCand3rdFDC}($\G, \*X, \*Y, \*I, \*R'$)}

    % \State {\bfseries Input:} \(\G\) a causal diagram; \(\*X,\*Y\) disjoint sets of variables; \(\*I,\*R\) sets of variables.
    
    \State {\bfseries Output:}
    % \sout{\(\*R'' \subseteq \*R'\) a set of candidate variables such that for every variable \(v \in \*R''\), there exists a subset \(\*Z'\) with \(\*Z' \subseteq \*R'' \setminus \{v\}\) where a union of sets \(\{v\} \cup \*Z'\)}
    \(\*R''\) consisting of all the variables $v\in \*R'$ such that there exists a set \(\*Z\) containing \(v\) with \(\*I \subseteq \*Z \subseteq \*R'\) that satisfies the third condition of the FD criterion relative to \((\*X,\*Y)\).

    % with while
    % \State \(\*R'' \gets \*R', \*F \gets \emptyset\)
    
    % \Do
    %     \If {\(\*F\) includes any \(v \in \*I\)}
    %         \State \textbf{return} \(\emptyset\)
    %     \EndIf
    %     \State \(\*R'' \gets \*R'' \setminus \*F, \*F \gets \emptyset\)
    %     \ForAll {\(v \in \*R''\)}
    %         \State \(\*Z' \gets \textsc{GetDep}(\G, \*X, \*Y, \{v\}, \*R'')\)
    %         \If {\(\*Z' = \perp\)}
    %             \State \(\*F \gets \*F \cup \{v\}\)
    %         \EndIf
    %     \EndFor
    % \doWhile {\(\*F \neq \emptyset\)} \label{func:getcand3rdfdc:endwhile}

    % \State \textbf{return} \(\*R''\)
    
    % without while
    \State \(\*R'' \gets \*R'\)
    
    % \ForAll {\(v \in \*R'\)}
    %     \State \(\*Z' \gets \textsc{GetDep}(\G, \*X, \*Y, \{v\}, \*R')\) \label{func:getcand3rdfdc:getdep}
    %     \If {\(\*Z' = \perp\)}
    %         \If {\(v \in \*I\)}
    %             \State \textbf{return} \(\emptyset\)
    %         \Else
    %             \State \(\*R'' \gets \*R'' \setminus \{v\}\)
    %         \EndIf
    %     \EndIf
    % \EndFor \label{func:getcand3rdfdc:endwhile}
    
    \State \textbf{for all} \(v \in \*R'\):
        \Indent
        \State \textbf{if} \(\textsc{GetDep}(\G, \*X, \*Y, \{v\}, \*R') = \perp\) \textbf{then:} \label{func:getcand3rdfdc:getdep}
            \Indent
            \State \textbf{if} \(v \in \*I\) \textbf{then:} \textbf{return} \(\perp\)
                % \Indent
                %     \State \textbf{return} \(\emptyset\)
                % \EndIndent
            \State \textbf{else:} \(\*R'' \gets \*R'' \setminus \{v\}\)
                % \Indent
                %     \State \(\*R'' \gets \*R'' \setminus \{v\}\)
                % \EndIndent
            \EndIndent
        \EndIndent
    \State \textbf{end for} \label{func:getcand3rdfdc:endwhile}
    
    \State \textbf{return} \(\*R''\)
\EndFunction
\end{algorithmic}
\caption{A function that outputs %a set of candidate variables  satisfying the third condition of front-door criterion.
the set of candidate variables potentially satisfying the second and third conditions of the FD criterion.}
\label{func:getcand3rdfdc}
% \end{figure}
% \end{minipage}
\end{wrapfigure}
In Step 2, \textsc{FindFDSet} calls the function \textsc{GetCand3rdFDC} presented in Fig.~\ref{func:getcand3rdfdc} to generate a set \(\*R''\) consisting of all the variables $v\in  \*R' $ such that there exists a set \(\*Z\) containing $v$ with \(\*I \subseteq \*Z \subseteq \*R'\) that further satisfies the third condition of the FD criterion relative to \((\*X,\*Y)\) (i.e., all BD paths from \(\*Z\) to \(\*Y\) are blocked by \(\*X\)). 
In other words, \(\*R'' \) is the union of all \(\*Z\) with \(\*I \subseteq \*Z \subseteq \*R'\) that satisfies the third condition of the FD criterion.

\textsc{GetCand3rdFDC} iterates through each variable \(v \in \*R'\) and calls the function \(\textsc{GetDep}(\G, \*X, \*Y, \{v\}, \*R')\) in line~\ref{func:getcand3rdfdc:getdep}.
Presented in Fig.~\ref{func:getdep}, \textsc{GetDep} returns a subset \(\*Z' \subseteq \*R' \setminus \{v\}\) such that all BD paths from \(\*Z = \{v\} \cup \*Z'\) to \(\*Y\) are blocked by \(\*X\) (if there exists such \(\*Z'\)).
%Then, we have that there exists a set \(\*Z \subseteq \*R'\), with \(v \in \*Z\), that satisfies the third condition of FD criterion relative to \((\*X,\*Y)\).
If \textsc{GetDep} returns \(\perp\), then there exists no \(\*Z\) containing $v$ that satisfies the third condition of the FD criterion relative to \((\*X,\*Y)\), so \(v\) is removed from \(\*R''\).

\begin{example} \label{eg-3rdFDC}
Continuing Example~\ref{eg-2ndFDC}. Given  \(\*I = \emptyset\) and \(\*R' = \{A,B,C\}\),  
\textsc{GetCand3rdFDC} outputs \(\*R'' = \{A,B,C\}\) because for each variable \(v \in \*R''\), \textsc{GetDep} finds a set \(\*Z'\) such that \(\{v\} \cup \*Z'\) satisfies the third condition of the FD criterion relative to \((\{X\},\{Y\})\). % where \(\*Z' \subseteq \*R' \setminus \{v\}\).
For \(v = A\), \(\*Z' = \emptyset\), 
for \(v = B\), \(\*Z' = \{A\}\), and for \(v = C\), \(\*Z' = \{A\}\).
\end{example}

Next, we explain how the function \textsc{GetDep}$(\G, \*X, \*Y, \*T, \*R')$ works. 
% First, \textsc{GetDep} prohibits all descendants of \(\*Y\) to be present in \(\*T\) and \(\*R'\).
% This is because, with \(\*Z' \subseteq \*R' \setminus \*T\) and \(\*Z = \*T \cup \*Z'\), any \(\*Z\) containing a descendant of $\*Y$ will violate the third condition of the FD criterion relative to \((\*X,\*Y)\).
% If any node \(v \in \De{\*Y}\) is in \(\*Z\), then there exists a BD path \(\pi\) from \(\*Z\) to \(\*Y\) %due to the BD path \(\pi\) from \(v\) to \(y \in \*Y\) 
% (i.e., the causal path from \(\*Y\) to \(v\)).
% Note that \(\*X\) cannot be blocking \(\pi\); otherwise, there will be a causal path from \(y\) to some \(x \in \*X\), which violates the assumption that \(\*X\) causes \(\*Y\) and not vice versa.
% Hence, \(\*Z\) violates the third condition of the FD criterion relative to \((\*X,\*Y)\).
% % Hence, \(\*Z\) violates the third condition of the FD criterion relative to \((\*X,\*Y)\) since there exists a BD path from \(\*Z\) to \(\*Y\) that cannot be blocked by \(\*X\).
First, \textsc{GetDep} constructs an undirected graph \(\mathcal{M}\) in a way that the paths from \(\*T\) to \(\*Y\) in \(\mathcal{M}\) represent all BD paths from \(\*T\) to \(\*Y\) that cannot be blocked by \(\*X\) in \(\G\).
% The auxiliary function \(\textsc{moralize}(\G)\) \textit{moralizes} a given graph \(\G\) into an undirected graph.
% \textsc{moralize} converts all directed edges of \(\G\) into undirected edges, and for every pair of nonadjacent nodes in \(\G\) that share a common child, \textsc{moralize} adds an undirected edge between such pair.
The auxiliary function \(\textsc{moralize}(\G)\) moralizes a given graph \(\G\) into an undirected graph.
The moralization is performed on the subgraph over \(\An{\*T \cup \*X \cup \*Y}\) instead of \(\G\) based on the following property: \(\*T\) and \(\*Y\) are $d$-separated by \(\*X\) in \(\G\) if and only if  \(\*X\) is a \(\*T\)-\(\*Y\) node cut (i.e., removing $\*X$ disconnects \(\*T\) and \(\*Y\)) in \(\G' = \textsc{moralize}(\G_{ \An{\*T \cup \*X \cup \*Y} })\) \citep{lauritzen:96}.

% \textcolor{red}{ Despite of an additional overhead to moralize a given graph (i.e., \(O(n^2)\) where \(n\) represents the number of nodes in a graph), the moralization process may be necessary in certain cases.
% For instance, the algorithms that list all minimal BD adjustment sets \citep{textor:11,zander:etal14} must moralize the input graph in order to call the algorithm (as a subroutine) that enumerates all minimal vertex separators of an undirected graph \citep{Takata2010}, which is the core of such listing algorithms.
% Please refer to the papers for more details.} \textcolor{blue}{---Delete. They come out of nowhere. The arguments does not show whether it's feasible or better to directly perform BFS on the original G.}
%In \textsc{GetDep}, moralization allows performing Breadth-First Search (BFS) from \(\*T\) to \(\*Y\) on \(\mathcal{M}\) while incrementally constructing 
\textsc{GetDep} performs Breadth-First Search (BFS) from \(\*T\) to \(\*Y\) on \(\mathcal{M}\) and incrementally constructs a subset \(\*Z' \subseteq \*R' \setminus \*T\) such that, after BFS terminates, there will be no BD path from \(\*Z = \*T \cup \*Z'\) to \(\*Y\) that cannot be blocked by \(\*X\) in \(\G\).
% While constructing \(\*Z'\), \textsc{GetDep} calls the function \(\textsc{GetNeighbors}(u, \mathcal{M})\) (presented in Fig.~\ref{func:getneighbors}, Appendix) to obtain all observed neighbors of \(u\) in \(\mathcal{M}\).
While constructing \(\*Z'\), \textsc{GetDep} calls the function \(\textsc{GetNeighbors}(u, \mathcal{M})\) (presented in Fig.~8, Appendix) to obtain all observed neighbors of \(u\) in \(\mathcal{M}\).

% Function - GetDep
% \begin{wrapfigure}{L}{0.8\textwidth}
\begin{figure}

\begin{algorithmic}[1]
\Function {GetDep} {$\G, \*X, \*Y, \*T, \*R'$}
% \Function {\textsc{GetDep}(\(\G, \*X, \*Y, \*T, \*R'\))}

    % \State {\bfseries Input:} \(\G\) a causal diagram; \(\*X,\*Y\) disjoint sets of variables; \(\*R''\) sets of variables.
    
    \State {\bfseries Output:} \(\*Z' \subseteq \*R' \setminus \*T\), a set of variables such that  \(\*T \cup \*Z'\) satisfies the third condition of the FD criterion relative to \((\*X,\*Y)\).
    
    % \State \textbf{if} \(\*T \cap \De{\*Y} \neq \emptyset\) \textbf{then:} \textbf{return} \(\perp\)
    % \State \(\*R' \gets \*R' \setminus \De{\*Y}\)
    
    % \State \(\G^L \gets \G\) with all bidirected edges \(A \leftrightarrow B\) replaced by a latent node \(L_{AB}\) and two edges \(L_{AB} \rightarrow A\) and \(L_{AB} \rightarrow B\)
    % \State \(\G' \leftarrow \G^L_{ \An{\*T \cup \*X \cup \*Y} }, \G'' \leftarrow \G'_{\underline{\*T}}\)
    
    % \State \(\G' \gets \G_{ \An{\*T \cup \*X \cup \*Y} }\) with all bidirected edges \(A \leftrightarrow B\) replaced by a latent node \(L_{AB}\) and two edges \(L_{AB} \rightarrow A\) and \(L_{AB} \rightarrow B\)
    \State \(\G' \gets \G_{ \An{\*T \cup \*X \cup \*Y} }\)
    \State \(\G' \gets \G'\) with all bidirected edges \(A \leftrightarrow B\) replaced by a latent node \(L_{AB}\) and two edges \(L_{AB} \rightarrow A\) and \(L_{AB} \rightarrow B\)
    \State \(\G'' \leftarrow \G'_{\underline{\*T}}\)
    \State \(\mathcal{M} \gets \textsc{moralize}(\G'')\) then remove \(\*X\)
    \State \(\*Z' \gets \emptyset, \*Q \gets \*T\) and mark all \(v \in \*T\) as visited
    % \State \(\*Q \gets \*T\)
    % \State mark all \(v \in \*T\) as visited
    
    \While {\(\*Q \neq \emptyset\)}
        \State \(u \gets \*Q.\textsc{pop}()\) \label{getdep:popfromqueue}
        
        % \If {\(u \in \*Y\)}
        %     \State \textbf{return} \(\perp\)
        % \EndIf
        \State \textbf{if} \(u \in \*Y\) \textbf{then:} \textbf{return} \(\perp\) \label{getdep:returnfail}
        
        \State \(\*{NR} \gets \textsc{GetNeighbors}(u,\mathcal{M}) \cap \*R'\) that are not visited
        \State \(\G'' \leftarrow \G'_{\underline{\*T \cup \*Z' \cup \*{NR}}}\)
        \State \(\mathcal{M} \gets \textsc{moralize}(\G'')\) then remove \(\*X\)
        \State \(\*N' \gets \textsc{GetNeighbors}(u,\mathcal{M})\) that are not visited
        \State \(\*{NR}' \gets \{w \in \*{NR} | \) there exists an incoming arrow into \(w\) in \(\G \} \)
        \State \(\*N \gets \*N' \cup \*{NR}', \*Z' \gets \*Z' \cup \*{NR}\)
        % \State \(\*Z' \gets \*Z' \cup \*{NR}\)
        \State \(\*Q.\textsc{insert}(\*N)\) and mark all \(w \in \*N\) as visited \label{getdep:inserttoqueue}
    \EndWhile
    
    \State \textbf{return} \(\*Z'\)
\EndFunction
\end{algorithmic}
\caption{A function that facilitates the construction of a set that satisfies the third condition of the FD criterion.}
\label{func:getdep}
\end{figure}
% \end{wrapfigure}

% \textsc{GetDep} performs Breadth-First Search (BFS) from \(\*T\) to \(\*Y\) in \(\mathcal{M}\) and incrementally constructs a subset \(\*Z' \subseteq \*R' \setminus \*T\) such that, after BFS terminates, there will be no BD path from \(\*Z = \*T \cup \*Z'\) to \(\*Y\) that cannot be blocked by \(\*X\) in \(\G\).
% Then, we have that all BD paths from \(\*Z\) to \(\*Y\) are blocked by \(\*X\), and thus \(\*Z\) satisfies the third condition of the FD criterion relative to \((\*X,\*Y)\).
% While constructing \(\*Z'\), \textsc{GetDep} calls the function \(\textsc{GetNeighbors}(u, \mathcal{M})\) (presented in Fig.~\ref{func:getneighbors} in Appendix) to obtain all observed neighbors of \(u\) in \(\mathcal{M}\).

The BFS starts from each variable \(v \in \*T\).
Whenever a non-visited node \(u\) is encountered, the set \(\*{NR}\), observed neighbors of \(u\) that belong to \(\*R'\), is computed.
\(\*{NR}\) can be added to \(\*Z'\) because removing all outgoing edges of \(\*{NR}\) may contribute to disconnecting some BD paths \(\Pi\) from \(\*T\) to \(\*Y\) that cannot be blocked by \(\*X\) in \(\G\).
In other words, in \(\G_{\underline{\*T \cup \*Z' \cup \*{NR}}}\), \(\Pi\) could be disconnected from \(\*T\) to \(\*Y\) where \(\Pi\) are not disconnected in \(\G_{\underline{\*T \cup \*Z'}}\).
After adding \(\*{NR}\) to \(\*Z'\), \(\mathcal{M}\) must be reconstructed in a way that reflects the setting where all outgoing edges of \(\*{NR}\) are removed.
BFS will be performed on such modified \(\mathcal{M}\).

\textsc{GetDep} checks if there exists any set of nodes \(\*N\) to be visited further.
\(\*N\) consists of two sets: 1) \(\*N'\), all observed neighbors of \(u\) that are still reachable from \(u\), even after removing all outgoing edges of \(\*{NR}\), and 2) \(\*{NR'} \subseteq \*{NR}\) where for every node \(w \in \*{NR}\), there exists an incoming arrow into \(w\) in \(\G\).
All nodes in \(\*{NR'}\) must be checked because there might exist some BD path \(\pi\) from \(w\) to \(y \in \*Y\) that cannot be blocked by \(\*X\) in \(\G\).
If \(\pi\) cannot be disconnected from \(w\) to \(y\), then the set \(\*Z\) will violate the third condition of the FD criterion relative to \((\*X,\*Y)\).

The BFS continues until either a node \(y \in \*Y\) is visited, or no more nodes can be visited.
If \textsc{GetDep} returns a set \(\*Z'\), then we have that all BD paths from \(\*T\) to \(\*Y\) that cannot be blocked by \(\*X\) in \(\G\) have been disconnected in \(\G_{\underline{\*Z}}\) while ensuring that there exists no BD path from \(\*Z\) to \(\*Y\) that cannot be blocked by \(\*X\) in \(\G\).
Therefore, \(\*Z\) satisfies the third condition of the FD criterion relative to \((\*X,\*Y)\).
% Otherwise, if \textsc{GetDep} returns \(\perp\) (i.e., \(y\) is visited), then there does not exist any \(\*Z\) including \(\*T\) that satisfies the third condition of the FD criterion relative to \((\*X,\*Y)\).
Otherwise, if \textsc{GetDep} returns \(\perp\) (i.e., \(y\) is visited), then there does not exist any \(\*Z\) containing \(\*T\) that satisfies the third condition of the FD criterion relative to \((\*X,\*Y)\).
% This is because there exists a BD path \(\pi\) from \(t \in \*T\) to \(y\) that cannot be blocked by \(\*X\) in \(\G\), even when for all variables \(w \in \*R'\) that intersect \(\pi\), all outgoing edges of \(w\) in \(\G\) were removed.
% In other words, removing outgoing edges of all \(w \in \*R'\) that intersect \(\pi\) did not disconnect \(\pi\) from \(t\) to \(y\).
This is because there exists a BD path \(\pi\) from \(t \in \*T\) to \(y\) that cannot be blocked by \(\*X\) in \(\G\); removing outgoing edges of all \(w \in \*R'\) that intersect \(\pi\) cannot disconnect \(\pi\) from \(t\) to \(y\).

\begin{example}
% \sout{ Expanding further from Example~\ref{eg-3rdFDC}.}
Expanding on Example~\ref{eg-3rdFDC} to show the use of function \textsc{GetDep}.
Consider the case when \(v = B\).
Then, \(\*Q = \*T = \{B\}\) and \(u = B\) is popped from \(\*Q\) at line~\ref{getdep:popfromqueue}.
We have \(\*{NR} = \{A\}, \*N' = \emptyset, \*{NR'} = \{A\}, \*N = \{A\}\), and \(\*Z' = \{A\}\).
Since \(\*N\) is inserted to \(\*Q\) at line~\ref{getdep:inserttoqueue}, \(u = A\) is popped from \(\*Q\) in the next iteration of while loop.
Then, \(\*{NR} = \emptyset, \*N' = \emptyset, \*{NR'} = \emptyset\), and \(\*N = \emptyset\).
Since \(\*Q\) is empty, the while loop terminates and \textsc{GetDep} returns \(\*Z' = \{A\}\).
\end{example}

\begin{example}
Illustrating the use of function \textsc{GetDep}.
Let \(\*I = \emptyset\), \(\*R' = \{B,C\}\), and \(v = B\).
\(\*Q = \*T = \{B\}\) and \(u = B\) is popped from \(\*Q\) at line~\ref{getdep:popfromqueue}.
\(\*{NR} = \emptyset, \*N' = \{A\}, \*{NR'} = \emptyset, \*N = \{A\}\), and \(\*Z' = \emptyset\).
Since \(\*N\) is inserted to \(\*Q\) at line~\ref{getdep:inserttoqueue}, \(u = A\) is popped from \(\*Q\) in the second iteration of while loop.
\(\*{NR} = \*{NR'} = \{C\}\), \(\*N' = \*N = \{C,D,Y\}\), \(\*Z' = \{C\}\), and \(\*Q = \{C,D,Y\}\).
On the third iteration, \(u = C\) is popped from \(\*Q\).
\(\*{NR} = \*{NR'} = \*N' = \*N = \emptyset\) and \(\*Q = \{D,Y\}\).
On the fourth iteration, \(u = D\) is popped from \(\*Q\).
\(\*{NR} = \*{NR'} = \*N' = \*N = \emptyset\) and \(\*Q = \{Y\}\).
Next, \(u = Y\) is popped from \(\*Q\).
Since \(u \in \{Y\}\), \textsc{GetDep} returns \(\perp\) at line~\ref{getdep:returnfail}.
There exists no set \(\*Z' \subseteq (\*R' \setminus \*T) = \{C\}\) such that \(\*T \cup \*Z'\) satisfies the third condition of the FD criterion relative to \((\{X\},\{Y\})\).
\end{example}

% Lemma - GetCand3rdFDC
\begin{lemma}[Correctness of \textsc{GetCand3rdFDC}]
\label{lemma:findfdset:3rd}
\textsc{GetCand3rdFDC}\((\G, \*X, \*Y, \*I, \*R')\) in Step 2 of Alg.~\ref{alg:findfdset} generates a set of variables \(\*R''\) where \(\*I \subseteq \*R'' \subseteq \*R'\).
\(\*R''\) consists of all and only variables \(v\) such that there exists a subset \(\*Z\) with \(\*I \subseteq \*Z \subseteq \*R'\) and \(v \in \*Z\) that satisfies the third condition of the FD criterion relative to \((\*X,\*Y)\).
Further, every set \(\*Z\) with \(\*I \subseteq \*Z \subseteq \*R\) that satisfies both the second and the third conditions of the FD criterion must be a subset of \(\*R''\).
\end{lemma}

% address that not every subset in R'' satisfies 3-FDC
\textbf{Remark:} Even though every set \(\*Z\) with \(\*I \subseteq \*Z \subseteq \*R'\) that satisfies the third condition of the FD criterion must be a subset of \(\*R''\),  \emph{not}  every subset \(\*Z \subseteq \*R''\) satisfies the third condition of the FD criterion, as illustrated by the following example.

\begin{example}
\label{example:getcand3rdfdc:limitation}
In Example~\ref{eg-3rdFDC}, \textsc{GetCand3rdFDC} outputs \(\*R'' = \{A,B,C\}\).
%Consider two sets \(\*Z_1 = \{A,B\}\) and \(\*Z_2 = \{B,C\}\). \(\*Z_1\) satisfies the third condition of FD criterion relative to \((\{X\},\{Y\})\). However, for \(\*Z_2\), there exists two BD paths from \(\*Z_2\) to \(\{Y\}\) that cannot be blocked by \(\{X\}\), which are \(\{C \leftarrow A \rightarrow D \rightarrow Y\}\) and \(\{B \leftarrow A \rightarrow D \rightarrow Y\}\). \(\*Z_2\) violates the third condition of FD criterion relative to \((\{X\},\{Y\})\).
However, for \(\*Z = \{B\}\), the BD path \(\{B \leftarrow A \rightarrow D \rightarrow Y\}\) is not blocked by \(\{X\}\); for \(\*Z = \{C\}\), the BD path \(\{C \leftarrow A \rightarrow D \rightarrow Y\}\) is not blocked by \(\{X\}\).
\end{example}

On the other hand,  we show that \(\*Z = \*R''\) itself satisfies the third condition of the FD criterion, as shown in the following. 
% Lemma - R''
\begin{lemma}
\label{lemma:findfdset:r}
\(\*R''\) generated by \textsc{GetCand3rdFDC} (in Step 2 of Alg.~\ref{alg:findfdset}) satisfies the third condition of the FD criterion, that is, all BD paths from \(\*R''\) to \(\*Y\) are blocked by \(\*X\).
\end{lemma}

\subsection*{Step 3 of \textsc{FindFDSet}}

Finally, in Step 3, \textsc{FindFDSet} looks for a set \(\*Z \subseteq \*R''\) that satisfies the first condition of the FD criterion relative to \((\*X, \*Y)\), that is, \(\*Z\) intercepts all causal paths from \(\*X\) to \(\*Y\).
To facilitate checking whether a set \(\*Z\) intercepts all causal paths from \(\*X\) to \(\*Y\), we introduce the concept of causal path graph defined as follows.

% Fig - causal path graph
\begin{wrapfigure}{R}{0.5\textwidth}
% \begin{figure}[t]
    \centering
    \null\hfill%
    % \begin{subfigure}{0.4\textwidth}
    \begin{subfigure}{0.2\textwidth}
        \includegraphics[width=\textwidth]{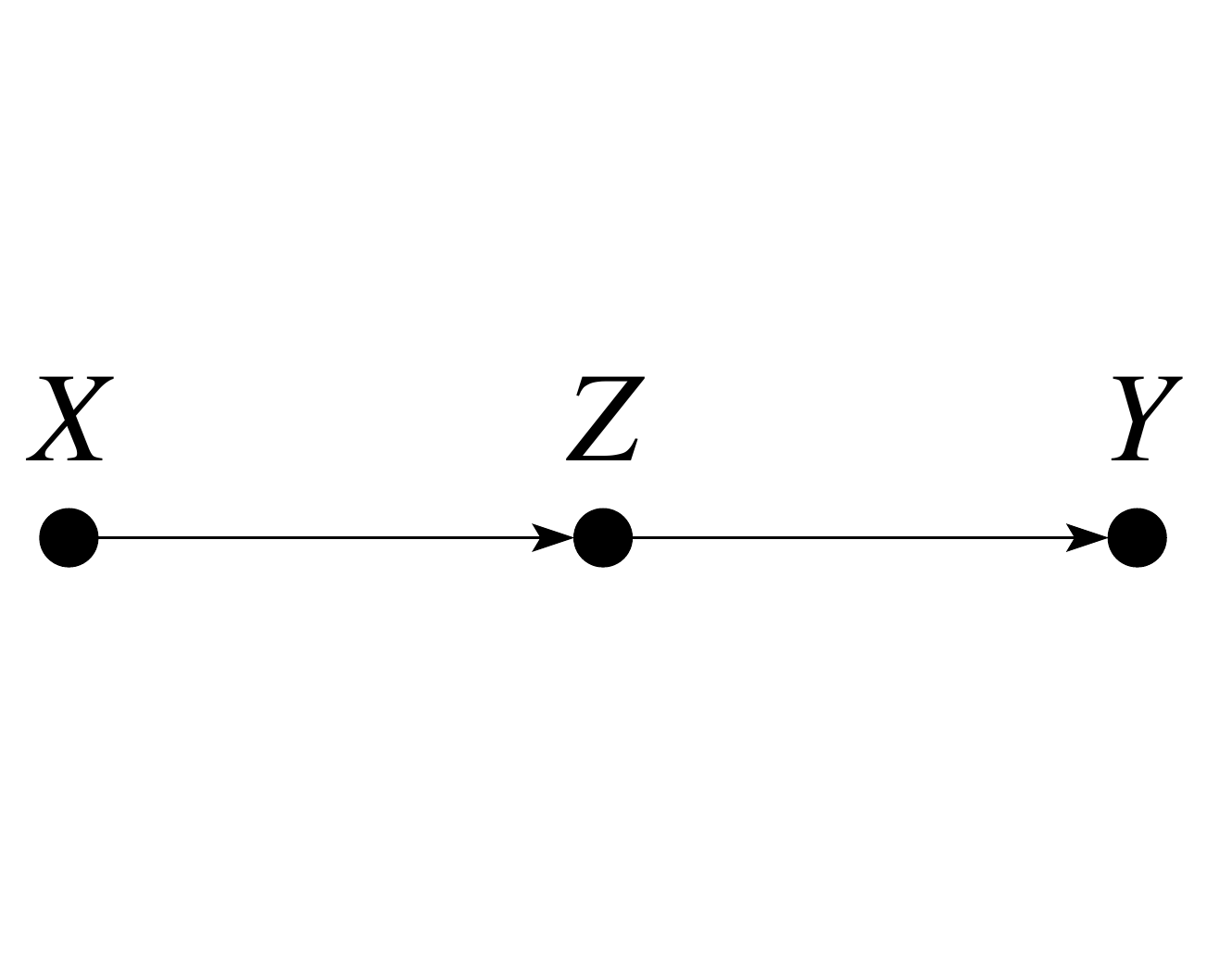}
        \caption{\(\G\)}
        \label{fig:fd:canonical:cpg}
    \end{subfigure}
    \hfill
    % \begin{subfigure}{0.4\textwidth}
    \begin{subfigure}{0.2\textwidth}
        \includegraphics[width=\textwidth]{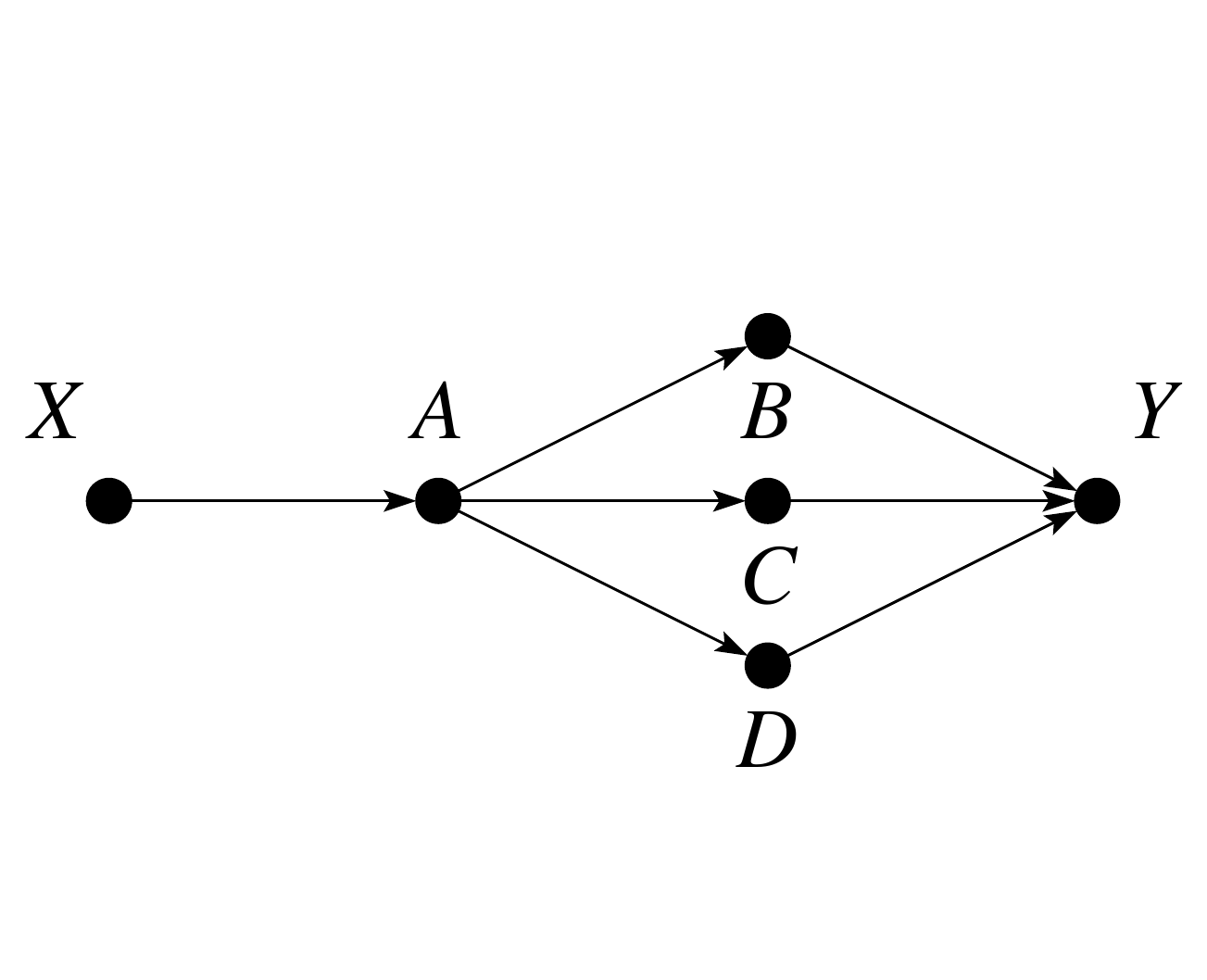}
        \caption{\(\G''\)}
        \label{fig:fd:intro:cpg}
    \end{subfigure}
    \null\hfill%
    % \begin{minipage}{0.23\textwidth}
    %     \centering
    %     \subfigure[\(\G_{4(a)}\)]{\label{fig:fd:canonical:cpg}
    %     \includegraphics[width=0.8\textwidth]{figures/fig_fd_canonical_cpg}}
    % \end{minipage}
    % \hfill
    % \begin{minipage}{0.23\textwidth}
    %     \centering
    %     \subfigure[\(\G_{4(b)}\)]{\label{fig:fd:intro:cpg}
    %     \includegraphics[width=0.8\textwidth]{figures/fig_fd_intro_cpg}}
    % \end{minipage}
\caption{
%Two causal path graphs, \(\G\) and \(\G'\) shown in (a) and (b), are generated from the original causal graphs \(\G\) and \(\G'\) (shown in Fig.~\ref{fig:fd:canonical} and Fig.~\ref{fig:fd:intro}) respectively.
Two causal path graphs generated from (a) the causal graph in Fig.~\ref{fig:fd:canonical}, and (b) the causal graph in Fig.~\ref{fig:fd:intro}.
% Both causal path graphs preserve all and only causal paths from \(\{X\}\) to \(\{Y\}\) in the original graphs.
Both preserve all and only causal paths from \(\{X\}\) to \(\{Y\}\) in the original graphs.
}
% \end{figure}
\end{wrapfigure}
% Def - causal path graph
\begin{definition}{(Causal Path Graph)}
\label{definition:causalpathgraph}
Let \(\G\) be a causal graph and \(\*X,\*Y\) disjoint sets of variables.
A \emph{causal path graph} \(\G'\) relative to \((\G,\*X,\*Y)\) is a graph over \(\*X \cup \*Y \cup PCP(\*X,\*Y)\), where \(PCP(\*X,\*Y) = (\De{\*X}_{\G_{\overline{\*X}}} \setminus \*X) \cap \An{\*Y}_{\G_{\underline{\*X}}}\)\footnote{A notation introduced by \citet{zander:etal14} to denote the set of variables on proper causal paths from \(\*X\) to \(\*Y\).}, constructed as follows:
\begin{enumerate}
    % \item \sout{Construct a subgraph \(G_S = \G_{\*X \cup \*Y \cup PCP(\*X,\*Y)}\).}
     \item Construct a subgraph \(\G''= \G_{\*X \cup \*Y \cup PCP(\*X,\*Y)}\).
   
    % \item \sout{ Transform the subgraph \(\G_S = \G_{\overline{\*X}\underline{\*Y}}\) and remove all bidirected edges of \(\G_S\).}
       \item Construct a graph \(\G'= \G''_{\overline{\*X}\underline{\*Y}}\), then  remove all bidirected edges from \(\G'\).
\end{enumerate}
\end{definition}

% We have developed a function \(\textsc{GetCausalPathGraph}(\G, \*X, \*Y)\), presented in  Fig.~\ref{func:getcausalpathgraph} in the Appendix, for 
% A function \(\textsc{GetCausalPathGraph}(\G, \*X, \*Y)\) for constructing a causal path graph is presented in Fig.~\ref{func:getcausalpathgraph} in the Appendix.
A function \(\textsc{GetCausalPathGraph}(\G, \*X, \*Y)\) for constructing a causal path graph is presented in Fig.~9 in the Appendix.
% We have developed a function \(\textsc{GetCausalPathGraph}(\G, \*X, \*Y)\) for constructing a causal path graph, which is presented in Fig.~\ref{func:getcausalpathgraph} in the Appendix.
% For constructing a causal path graph, we have developed a function \(\textsc{GetCausalPathGraph}(\G, \*X, \*Y)\) which is presented in Fig.~\ref{func:getcausalpathgraph} in the Appendix.

%For constructing a causal path graph, we have developed a function \(\textsc{GetCausalPathGraph}(\G, \*X, \*Y)\), presented in  Fig.~\ref{func:getcausalpathgraph} in the Appendix.
% We have developed a function \(\textsc{GetCausalPathGraph}(\G, \*X, \*Y)\), presented in  Fig.~\ref{func:getcausalpathgraph} in the Appendix, for constructing a causal path graph.
%First, we obtain a \textit{causal path graph} \(\G'\) by calling the function \(\textsc{GetCausalPathGraph}(\G, \*X, \*Y)\) (shown in Appendix, Fig.~\ref{func:getcausalpathgraph}).
\begin{example}
Consider the causal graph \(\G'\) shown in Fig.~\ref{fig:fd:intro} with \(\*X = \{X\}\), and \(\*Y = \{Y\}\).
The causal path graph \(\G''\) relative to \((\G',\{X\},\{Y\})\) is shown in Fig.~\ref{fig:fd:intro:cpg}.
All causal paths from \(\{X\}\) to \(\{Y\}\) in \(\G'\) are present in \(\G''\).
\end{example}

%We formally define a causal path graph as follows.

% After constructing a causal path graph \(\G'\) relative to \((\G,\*X,\*Y,\*Z)\), the function \(\textsc{TestSep}(\G', \*X, \*Y, \*Z)\) is called to check if \(\*Z\) is a separator of \(\*X\) and \(\*Y\) in \(\G'\).
%After constructing a causal path graph \(\G'\) relative to \((\G,\*X,\*Y)\),  the function \(\textsc{TestSep}(\G', \*X, \*Y, \*Z)\) is called to check if \(\*Z\) is a separator of \(\*X\) and \(\*Y\) in \(\G'\). If \textsc{TestSep} returns True, then \(\*Z\) satisfies the first condition of FD criterion relative to \((\*X,\*Y)\) based on the following lemma, and we have that \(\*Z\) satisfies the FD criterion relative to \((\*X,\*Y)\).
After constructing a causal path graph \(\G'\) relative to \((\G,\*X,\*Y)\), we use the function \(\textsc{TestSep}(\G', \*X, \*Y, \*Z)\) to check if \(\*Z\) is a separator of \(\*X\) and \(\*Y\) in \(\G'\).  Based on the following lemma, \(\*Z\) satisfies the first condition of the FD criterion relative to \((\*X,\*Y)\) if and only if \textsc{TestSep} returns True.
% Lemma - FindFDSet - 1st
\begin{lemma}
\label{lemma:findfdset:separator}
Let \(\G\) be a causal graph and \(\*X,\*Y,\*Z\) disjoint sets of variables.
Let \(\G'\) be the causal path graph relative to \((\G,\*X,\*Y)\).
Then, \(\*Z\) satisfies the first condition of the FD criterion relative to \((\*X, \*Y)\) if and only if \(\*Z\) is a separator of \(\*X\) and \(\*Y\) in \(\G'\).
\end{lemma}

%\textsc{FindFDSet} serves an another purpose besides its main one, which is to obtain any FD adjustment set \(\*Z\) under the constraint \(\*I \subseteq \*Z \subseteq \*R\). That is, \textsc{FindFDSet} allows one to quickly verify that there does not exist any FD adjustment set \(\*Z\) with \(\*I \subseteq \*Z \subseteq \*R\) by checking that \textsc{FindFDSet} returns \(\perp\). The result is stated in the following lemma.
Given the set \(\*R''\) that contains every set \(\*Z\) with \(\*I \subseteq \*Z \subseteq \*R\) that satisfies both the second and the third conditions of the FD criterion (Lemma~\ref{lemma:findfdset:3rd}), it may appear that we need to search for a set \(\*Z \subseteq \*R''\) that satisfies the first condition of the FD criterion. We show instead that   all we need is to check whether the set $\*R''$ itself satisfies the first condition which has been shown to satisfy the second and third conditions by Lemma~\ref{lemma:findfdset:r}. This result is summarized in the following lemma.
% Lemma - exists FD
\begin{lemma}
\label{lemma:findfdset:exists}
%Let \(\G\) be a causal graph, \(\*X,\*Y\) disjoint sets of variables, and \(\*I, \*R\) sets of variables with \(\*I \subseteq \*R\) and \(\*R \cap (\*X \cup \*Y) = \emptyset\).
There exists a set \(\*Z_0\) satisfying the FD criterion relative to \((\*X,\*Y)\) with \(\*I \subseteq \*Z_0 \subseteq \*R\) if and only if \(\*R''\) generated by \textsc{GetCand3rdFDC} (in Step 2 of Alg.~\ref{alg:findfdset}) satisfies the FD criterion relative to \((\*X,\*Y)\).
\end{lemma}

\begin{example} \label{eg-lastFDC}
Continuing Example~\ref{eg-3rdFDC}. In Step 3, \textsc{FindFDSet} outputs \(\*Z = \*R'' = \{A,B,C\}\) since $\*Z$ is a separator of $\{X\}$ and $\{Y\}$ in the causal path graph \(\G''\) in Fig.~\ref{fig:fd:intro:cpg}.
\end{example}

%The result is summarized in the following lemma.
The results in this section are summarized as follows. 
% Correctness - FindFDSet
\begin{theorem}
[Correctness of \textsc{FindFDSet}]
\label{thm:findfdset}
Let \(\G\) be a causal graph, \(\*X,\*Y\) disjoint sets of variables, and \(\*I, \*R\) sets of variables such that \(\*I \subseteq \*R\).
Then, \textsc{FindFDSet}\((\G, \*X, \*Y, \*I, \*R)\) outputs a %front-door adjustment 
set \(\*Z\) with \(\*I \subseteq \*Z \subseteq \*R\) that satisfies the FD criterion relative to \((\*X,\*Y)\), or outputs \(\perp\) if none exists, in \(O(n^3 (n+m))\) time, where \(n\) and \(m\) represent the number of nodes and edges in \(\G\).
\end{theorem}

%% file: 4.fd_algorithm.tex
\section{Enumerating Front-door Adjustment Sets}
\label{section:fd_algorithm}

% Algorithm - ListFDSets
\begin{wrapfigure}{R}{0.53\textwidth}
\begin{minipage}[t]{0.53\textwidth}
\begin{algorithm}[H]
\caption{\textsc{ListFDSets} (\(\G, \*X, \*Y, \*I, \*R\))}
\label{alg:listfdsets}

\begin{algorithmic} [1]
    \State {\bfseries Input:} \(\G\) a causal diagram; \(\*X,\*Y\) disjoint sets of variables; \(\*I,\*R\) sets of variables.
    
    \State {\bfseries Output:} Listing front-door adjustment set \(\*Z\) relative to \((\*X,\*Y)\) where \(\*I \subseteq \*Z \subseteq \*R\).

    % \If {\(\textsc{FindFDSet}(\G, \*X, \*Y, \*I, \*R) \neq \perp\) } \label{alg:listfdsets:findfdset}
    %     \If {\(\*I = \*R\)}
    %         \State Output \(\*I\)
    %     \Else
    %         \State \(v \gets \) any variable from \(\*R \setminus \*I\) \label{alg:listfdsets:branchstart}
    %         \State \(\textsc{ListFDSets}(\G, \*X, \*Y, \*I \cup \{v\}, \*R)\) \label{alg:listfdsets:branchleft}
    %         \State \(\textsc{ListFDSets}(\G, \*X, \*Y, \*I, \*R \setminus \{v\})\) \label{alg:listfdsets:branchright}
    %     \EndIf
    % \EndIf
    
    % \State \textbf{if} \(\*I \cap \De{\*Y} \neq \emptyset\) \textbf{then:}
    % \textbf{return}
    % \State \(\*R \gets \*R \setminus \De{\*Y}\)
    
    \State \textbf{if} \(\textsc{FindFDSet}(\G, \*X, \*Y, \*I, \*R) \neq \perp\) \textbf{then:} \label{alg:listfdsets:findfdset}
        \Indent
        \State \textbf{if} \(\*I = \*R\) \textbf{then:} Output \(\*I\)
            % \Indent
            % \State Output \(\*I\)
            % \EndIndent
        \State \textbf{else:}
            \Indent
            \State \(v \gets \) any variable from \(\*R \setminus \*I\) \label{alg:listfdsets:branchstart}
            \State \(\textsc{ListFDSets}(\G, \*X, \*Y, \*I \cup \{v\}, \*R)\) \label{alg:listfdsets:branchleft}
            \State \(\textsc{ListFDSets}(\G, \*X, \*Y, \*I, \*R \setminus \{v\})\) \label{alg:listfdsets:branchright}
            \EndIndent
        \EndIndent
\end{algorithmic}
\end{algorithm}
\end{minipage}
\end{wrapfigure}
Our goal in this section is to develop an algorithm that lists \textit{all} FD adjustment sets in a causal diagram.
In general, there may exist exponential number of such sets, which means that any listing algorithm will take exponential time to list them all.
We will instead look for an algorithm that has an interesting property known as  \textit{polynomial delay} \citep{Takata2010}.
% In words, poly-delay algorithms are required to output the first answer (or indicate none is available) in polynomial time, and take polynomial time to output each consecutive answer as well.
In words, poly-delay algorithms output the first answer (or indicate none is available) in polynomial time, and take polynomial time to output each consecutive answer as well.
Consider the following example.

\begin{example}
\label{ex:fd:exp}
Consider the three causal graphs in Fig.~\ref{fig:fd:exp}.
In \(\G\) shown in Fig.~\ref{fig:fd:exp:1}, there exists 9 valid FD adjustment sets relative to \((\{X\},\{Y\})\).
In \(\G'\), presented in Fig.~\ref{fig:fd:exp:2}, two variables \(A_3\) and \(B_3\) are added from \(\G\), forming an additional causal path from \(X\) to \(Y\).
27 FD adjustment sets relative to \((\{X\},\{Y\})\) are available in \(\G'\).
If another causal path \(X \rightarrow A_4 \rightarrow B_4 \rightarrow Y\) is added to \(\G'\), then there are 81 FD adjustment sets relative to \((\{X\},\{Y\})\).
As shown in Fig.~\ref{fig:fd:exp:3}, in a graph \(\G''\) with similar pattern with causal path \(X \rightarrow A_i \rightarrow B_i \rightarrow Y, i=1, \ldots n\), there are at least $3^n$ number of FD adjustment sets.
%When two more nodes are added to the graph with an identical pattern, the total number of FD adjustment sets \(\mathcal{C}\) triples.
%With \(n\) the number of nodes on a causal graph, we have \(\mathcal{C} = 9\) with \(n = 4\), \(\mathcal{C} = 27\) with \(n = 6\) and \(\mathcal{C} = 81\) with \(n = 8\). \(\mathcal{C}\) grows exponentially with respect to \(n\).
\end{example}

% Given \(n\) the number of nodes on a causal graph \(\G\), the total number of valid FD adjustment sets is: \(\mathcal{C} = \sum_{k=0}^{n} \binom{n}{k} - \mathcal{F}\) where \(\mathcal{F}\) stands for the total number of cases that a set of size \(k\) with \(0 \leq k \leq n\) is not a FD adjustment set.
% Although it is possible to have a case without any valid FD adjustment set, in terms of the total number of valid sets with respect to \(n\), it is clear that \(\mathcal{C}\) grows exponentially.

% Ex - exponential number of FD sets
\begin{figure}[t]
    \centering
    
    % \begin{minipage}{0.3\textwidth}
    %     \centering
    %     \subfigure[\(\G_{5(a)}\)]{\label{fig:fd:exp:1}
    %     \includegraphics[width=\textwidth]{figures/fig_fd_exp_1}}
    % \end{minipage}
    % \hfill
    % \begin{minipage}{0.3\textwidth}
    %     \centering
    %     \subfigure[\(\G_{5(b)}\)]{\label{fig:fd:exp:2}
    %     \includegraphics[width=\textwidth]{figures/fig_fd_exp_2}}
    % \end{minipage}
    % \hfill
    % \begin{minipage}{0.3\textwidth}
    %     \centering
    %     \subfigure[\(\G_{5(c)}\)]{\label{fig:fd:exp:3}
    %     \includegraphics[width=\textwidth]{figures/fig_fd_exp_3}}
    % \end{minipage}
    
    \begin{subfigure}{0.25\textwidth}
        \includegraphics[width=\textwidth]{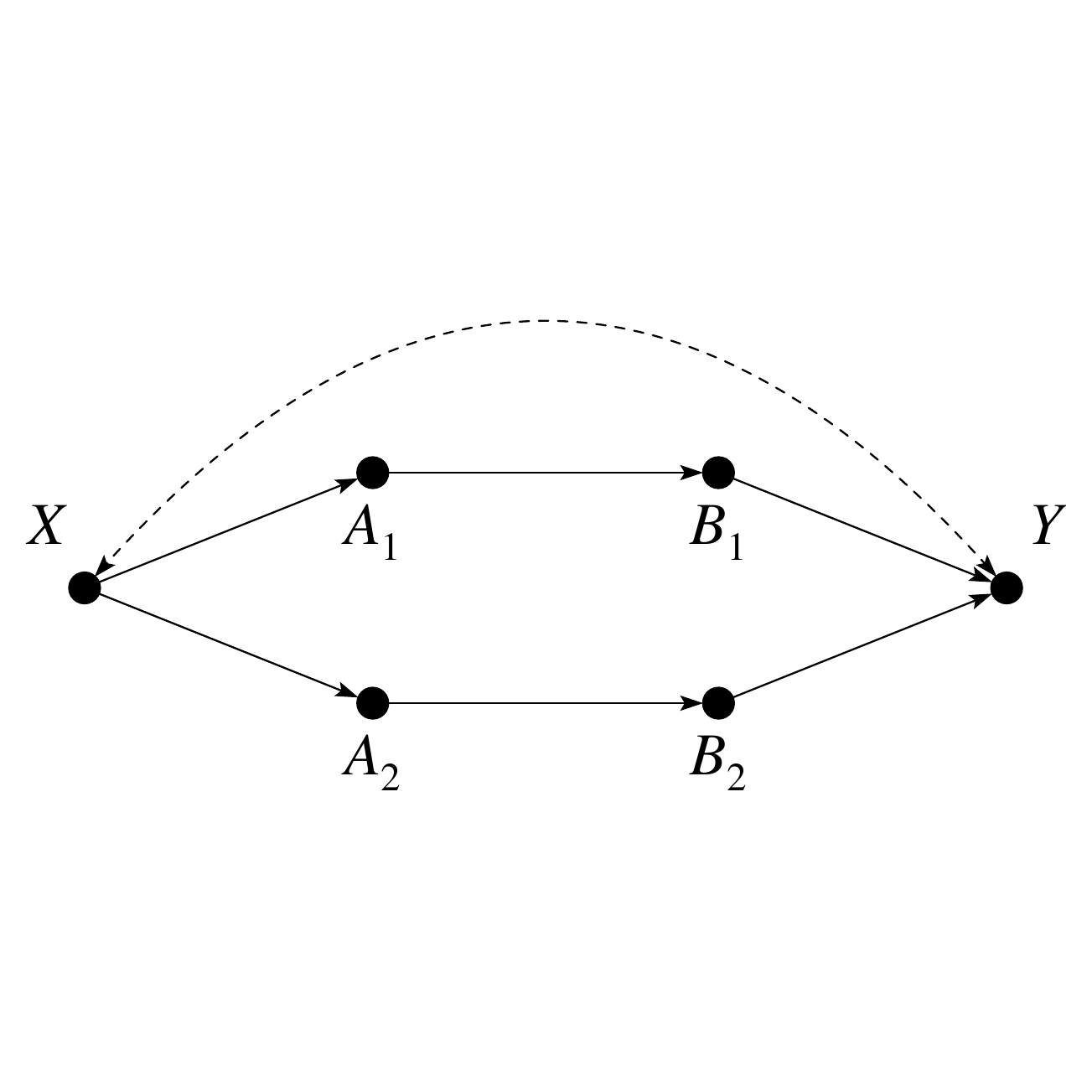}
        \caption{\(\G\)}
        \label{fig:fd:exp:1}
    \end{subfigure}
    \hfill
    \begin{subfigure}{0.25\textwidth}
        \includegraphics[width=\textwidth]{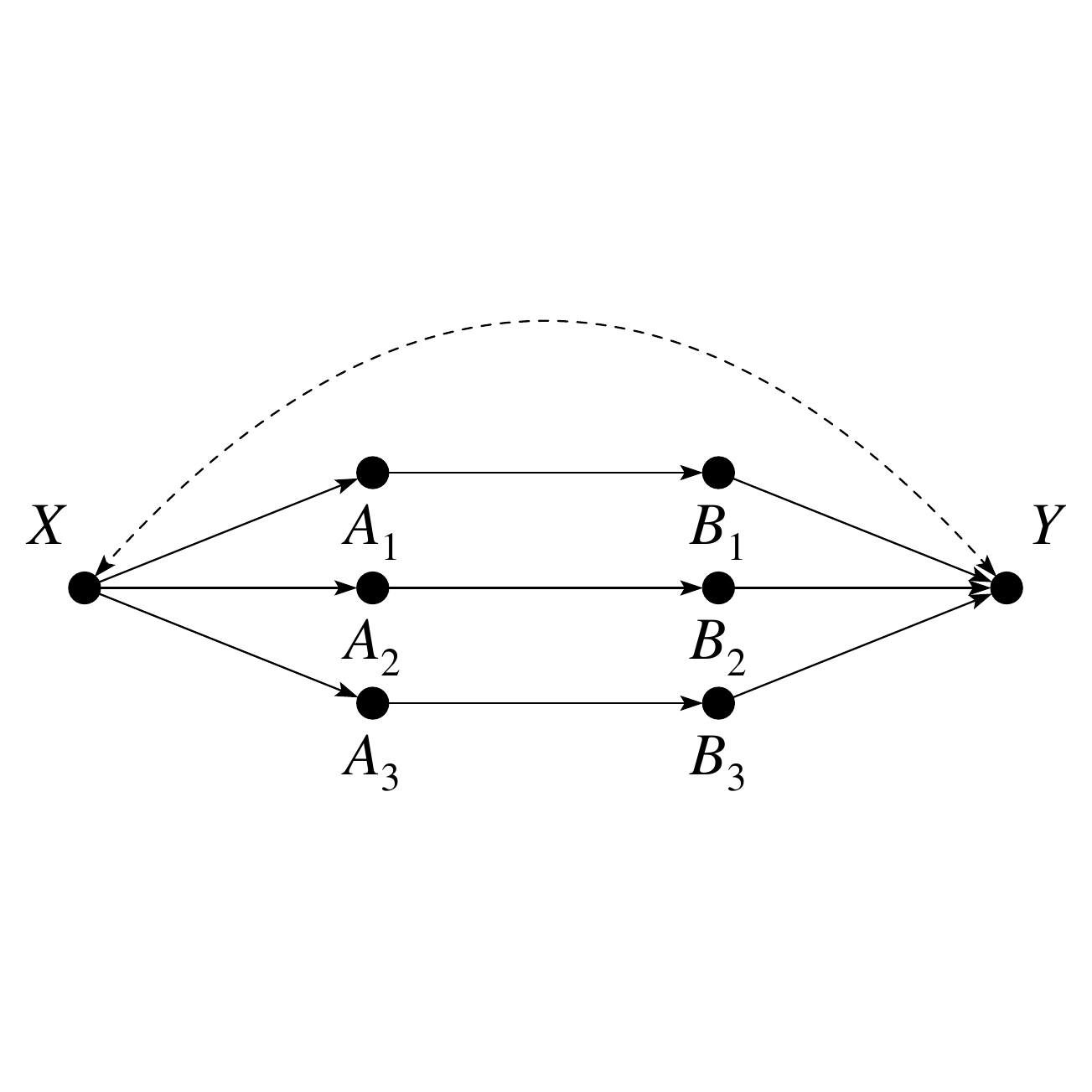}
        \caption{\(\G'\)}
        \label{fig:fd:exp:2}
    \end{subfigure}
    \hfill
    \begin{subfigure}{0.3\textwidth}
        \includegraphics[width=\textwidth]{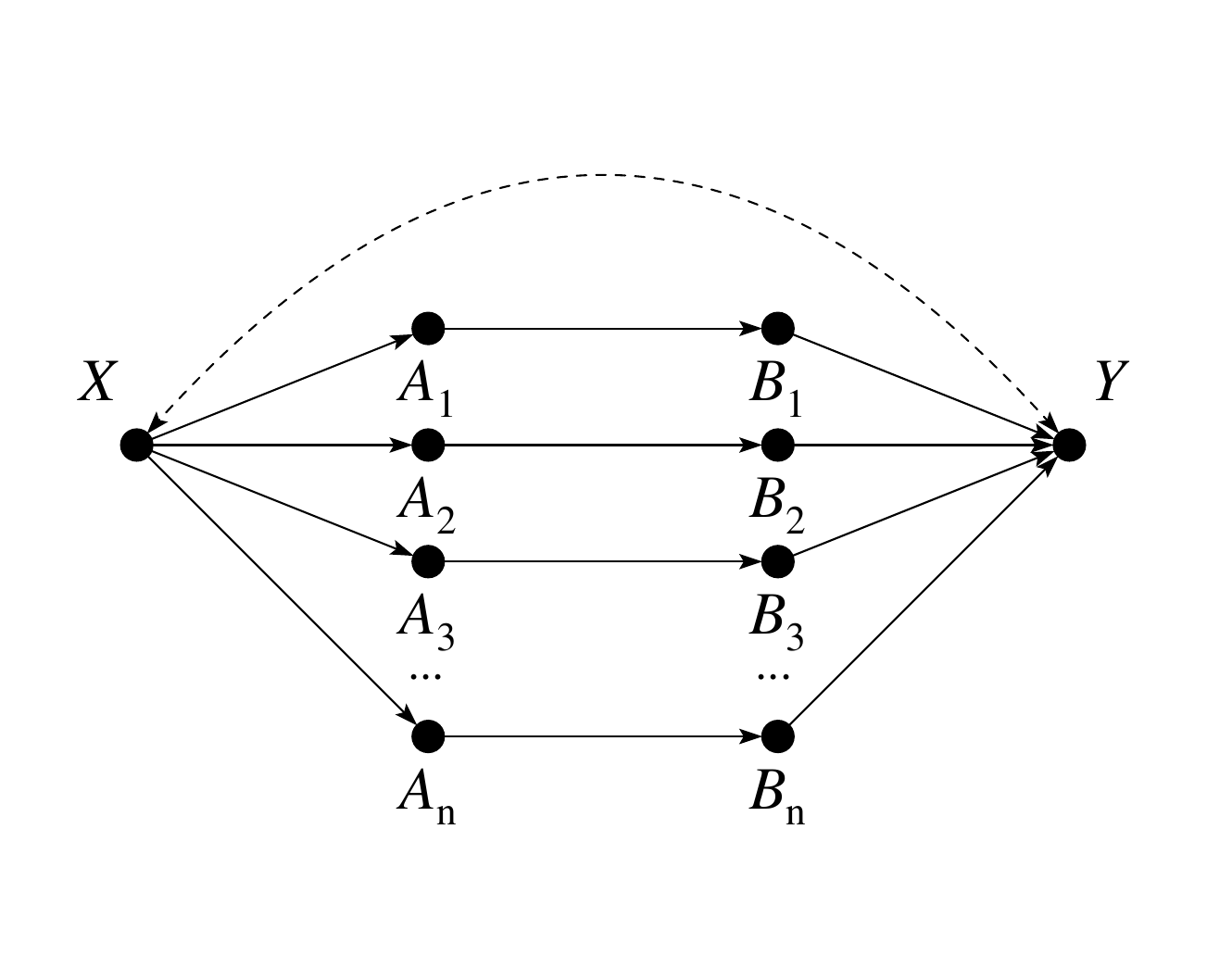}
        \caption{\(\G''\)}
        \label{fig:fd:exp:3}
    \end{subfigure}
\caption{
Three examples of the FD criterion to demonstrate that total number of FD adjustment sets may be exponential with respect to the number of nodes in a graph.
% Three examples of the FD criterion are shown.
% The total number of FD adjustment sets is exponential with respect to the number of nodes on a graph.
}
\label{fig:fd:exp}
\end{figure}

We have developed an algorithm named \textsc{ListFDSets}, shown in Alg.~\ref{alg:listfdsets}, that lists all FD adjustment sets $\*Z$ relative to \((\*X, \*Y)\) satisfying \(\*I \subseteq \*Z \subseteq \*R\) with polynomial delay, given a causal diagram \(\G\), disjoint sets of variables \(\*X\) and \(\*Y\), and two sets of variables \(\*I\) and \(\*R\).

\begin{example}
\label{ex:listfdsets}
Consider the causal graph \(\G'\) shown in Fig.~\ref{fig:fd:intro} with \(\*X = \{X\}\), \(\*Y = \{Y\}\), \(\*I = \emptyset\) and \(\*R = \{A,B,C,D\}\).
\textsc{ListFDSets} outputs \(\{A,B,C\}\),  \(\{A,B\}\), \(\{A,C\}\), \(\{A\}\) one by one,  % upon each subsequent call. On the fifth call, \textsc{ListFDSets} outputs $\perp$ 
and finally stops as no more  adjustment sets exist.
\end{example}

The algorithm \textsc{ListFDSets} takes the same search strategy as the listing algorithm \textsc{ListSep} \citep{zander:etal14} that enumerates all BD adjustment sets with polynomial delay. 
\textsc{ListFDSets} implicitly constructs a binary search tree where each tree node \(\mathcal{N}(\*I', \*R') \) represents the collection of all FD adjustment sets \(\*Z\) relative to \((\*X,\*Y)\) with \(\*I' \subseteq \*Z \subseteq \*R'\).
The search starts from the root tree node \(\mathcal{N}(\*I, \*R) \), indicating that \textsc{ListFDSets} will list all FD adjustment sets \(\*Z\) relative to \((\*X,\*Y)\) with \(\*I \subseteq \*Z \subseteq \*R\).
%In the case of \textsc{ListSep}, \(\mathcal{N}\) contains all BD adjustment sets \(\*Z\) relative to \((\*X,\*Y)\) where \(\*I' \subseteq \*Z \subseteq \*R'\).

%Upon visiting any \(\mathcal{N}\), both algorithms, \textsc{ListFDSets} and \textsc{ListSep}, call different functions, \textsc{FindFDSet} and \textsc{FindSep} respectively, that serve the same purpose \citep{zander:etal14}. \textsc{FindFDSet}, shown in Alg.~\ref{alg:listfdsets} line 3, outputs a FD adjustment set \(\*Z\) relative to \((\*X,\*Y)\) if there exists a FD adjustment set \(\*Z_0\) with \(\*I' \subseteq \*Z_0 \subseteq \*R'\), whereas \textsc{FindSep} outputs a BD adjustment set \(\*W\) relative to \((\*X,\*Y)\) if there exists a BD adjustment set \(\*W_0\) with \(\*I' \subseteq \*W_0 \subseteq \*R'\).

% Fig - ListFDSets
\begin{wrapfigure}{R}{0.5\textwidth}
    \centering
    \includegraphics[width=0.5\textwidth]{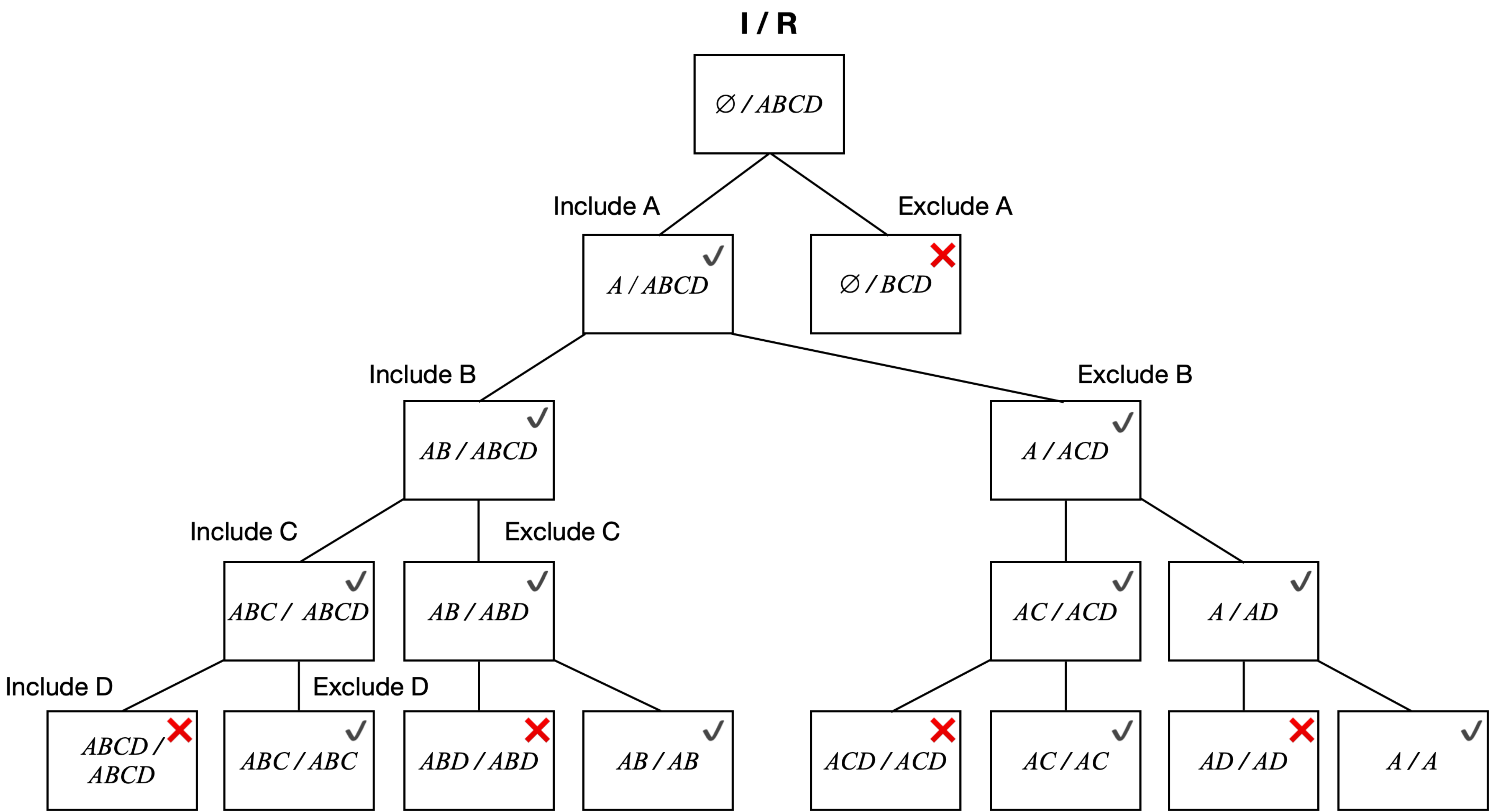}
\caption{
A search tree illustrating the running of \textsc{ListFDSets} in Example~\ref{ex:listfdsets:searchtree}.
}
\label{fig:listfdsets:searchtree}
% \vspace{-0.55in}
\vspace{-0.2in}
\end{wrapfigure}
Upon visiting a node \(\mathcal{N}(\*I', \*R') \), \textsc{ListFDSets} first calls the  function \textsc{FindFDSet} (line~\ref{alg:listfdsets:findfdset}) to decide whether it is necessary to search further from \(\mathcal{N}\).
If \textsc{FindFDSet} outputs \(\perp\), then there does not exist any FD adjustment set \(\*Z_0\) with \(\*I' \subseteq \*Z_0 \subseteq \*R'\) and there is no need to search further. 
% (i.e., the subtree \(\mathcal{T}\) rooted at \(\mathcal{N}\) is a barren: \(\mathcal{T}\) consists of \(\mathcal{N}\) only).
% Otherwise, \(\mathcal{T}\) branches out into two children, \(\mathcal{N}_1\) and \(\mathcal{N}_2\), and \textsc{ListFDSets} continues the search over each child separately.
Otherwise, \(\mathcal{N}\) spawns two children, \(\mathcal{N}_1\) and \(\mathcal{N}_2\), and \textsc{ListFDSets} continues the search over each child separately.
\(\mathcal{N}_1\) in line~\ref{alg:listfdsets:branchleft} represents the collection of all FD adjustment sets \(\*Z_1\) relative to \((\*X,\*Y)\) where \(\*I' \cup \{v\} \subseteq \*Z_1 \subseteq \*R'\).
On the other hand, \(\mathcal{N}_2\) in line~\ref{alg:listfdsets:branchright} represents the collection of all FD adjustment sets \(\*Z_2\) where \(\*I' \subseteq \*Z_2 \subseteq \*R' \setminus \{v\}\).
% \(\mathcal{N}_1\) and \(\mathcal{N}_2\) are disjoint: every FD adjustment set in \(\mathcal{N}_1\) includes \(v\) where none of the sets in \(\mathcal{N}_2\) does.
% The construction of \(\mathcal{N}_1\) and \(\mathcal{N}_2\) guarantees that the search never overlaps, which is crucial to guaranteeing that \textsc{ListFDSets} runs in polynomial delay.
\(\mathcal{N}_1\) and \(\mathcal{N}_2\) are disjoint and thus the search never overlaps, which is crucial to guaranteeing that \textsc{ListFDSets} runs in polynomial delay.
%\textsc{FindSep} behaves identically as \textsc{FindFDSet} with BD counterparts.
%When a leaf tree node \(\mathcal{L}\) is reached (i.e., when \(\*I' = \*R'\)), \textsc{ListFDSets} guarantees that \(\mathcal{L}\) includes a valid FD adjustment set \(\*I'\) (if \(\*I'\) is indeed FD-admissible). When \(\mathcal{L}\) is visited, \textsc{ListFDSets} calls the function \textsc{FindFDSet} where the arguments \(\*I\) and \(\*R\) are identical to \(\*I'\). \textsc{FindFDSet} will output \(\*I'\) if and only if \(\*I'\) satisfies the FD criterion relative to \((\*X,\*Y)\).
%Similarly, \textsc{ListSep} guarantees that a valid BD adjustment set \(\*I'\) relative to \((\*X,\*Y)\), if such \(\*I'\) exists, will be outputted by \textsc{FindSep} when \(\mathcal{L}\) is reached.
Finally, a leaf tree node \(\mathcal{L}\) is reached when \(\*I' = \*R'\), and  \textsc{ListFDSets} outputs a valid FD adjustment set \(\*I'\).

% \noindent
% \begin{minipage}{0.5\textwidth}
%   TEXT 1
% \end{minipage}
% \begin{minipage}{0.5\textwidth}
% % Fig - ListFDSets
% \begin{figure}[t]
%     \centering
%     \includegraphics[width=0.5\textwidth]{figures/fig_listfdsets}
% \caption{
% A search tree illustrating the running of \textsc{ListFDSets} in Example~\ref{ex:listfdsets:searchtree}.
% }
% \label{fig:listfdsets:searchtree}
% \end{figure}
% \end{minipage}

\begin{example}
\label{ex:listfdsets:searchtree}
Continuing from Example~\ref{ex:listfdsets}.
 Fig.~\ref{fig:listfdsets:searchtree} shows a search tree generated by running \textsc{ListFDSets}\((\G', \{X\}, \{Y\}, \emptyset, \{A,B,C,D\})\). 
Initially, the search starts from the root tree node \(\mathcal{N}(\emptyset, \{A,B,C,D\})\).
Since \textsc{FindFDSet} returns a set \(\{A,B,C\}\), \(\mathcal{N}\) branches out into two children \(\mathcal{N}'(\{A\}, \{A,B,C,D\})\) and \(\mathcal{N}''(\emptyset, \{B,C,D\})\).
The search continues from the left child \(\mathcal{N}'\) until reaching the leaf tree node \(\mathcal{L}_1(\{A,B,C,D\}, \{A,B,C,D\})\) where \textsc{FindFDSet} returns \(\perp\).
\textsc{ListFDSets} backtracks to the parent tree node \(\mathcal{N}_1(\{A,B,C\}, \{A,B,C,D\})\) and then checks the next leaf \(\mathcal{L}_2(\{A,B,C\}, \{A,B,C\})\) where \textsc{FindFDSet} returns a set \(\{A,B,C\}\), a valid FD admissible set relative to \((\{X\},\{Y\})\).
\textsc{ListFDSets} outputs \(\{A,B,C\}\).
Next, \textsc{ListFDSets} backtracks to the tree node \(\mathcal{N}_2(\{A,B\}, \{A,B,C,D\})\) and reaches the leaf \(\mathcal{L}_3(\{A,B\}, \{A,B\})\) where \textsc{FindFDSet} outputs \(\{A,B\}\), and thus \textsc{ListFDSets} outputs \(\{A,B\}\).
\textsc{ListFDSets} continues and outputs two sets \(\{A,C\}\) and \(\{A\}\) in order.
Finally, \textsc{ListFDSets} backtracks to the root \(\mathcal{N}\) and checks the right child \(\mathcal{N}''\) where \textsc{FindFDSet} returns \(\perp\).
\textsc{ListFDSets} does not search further from \(\mathcal{N}''\) and stops as no more tree node is left to be visited.
\end{example}

%% file: 5.fd_theorem.tex
% Soundness & Completeness

Our results are summarized in the following theorem, which provides the correctness, completeness, and poly-delay complexity of the proposed algorithm.
% \textcolor{red}{Clarified the term ``completeness''.}
% \sout{Note that the completeness of the algorithm is different from the completeness of Pearl's FD criterion itself.
% The algorithm is complete, which means that it lists all and only all valid sets satisfying the FD criterion.} 
Note that the completeness of the algorithm means that it lists ``all'' valid sets satisfying the FD criterion.
On the other hand, Pearl's FD criterion is not complete in the sense that there might exist a causal effect that can be computed by the FD adjustment formula (Eq.~(\ref{eq:fd})) but the set \(\*Z\) does not satisfy the FD criterion.

\begin{theorem}[Correctness of \textsc{ListFDSets}]
\label{thm:listfdsets}
Let \(\G\) be a causal graph, \(\*X,\*Y\) disjoint sets of variables, and \(\*I,\*R\) sets of variables.
\textsc{ListFDSets\((\G,\*X,\*Y, \*I, \*R)\)} enumerates all and only sets \(\*Z\) with \(\*I \subseteq \*Z \subseteq \*R\) that satisfy the FD criterion relative to \((\*X, \*Y)\) in \(O(n^4 (n+m))\) delay where \(n\) and \(m\) represent the number of nodes and edges in \(\G\).
\end{theorem}

% % Time Complexity
% \begin{theorem}[Complexity of \textsc{ListFDSets}]
% \label{theorem:listfdset:complexity}
% \textsc{ListFDSets} runs with \(O(n^4 (n+m))\) delay where \(n\) and \(m\) represent the number of nodes and edges in \(\G\).
% \end{theorem}

%% file: 6.limitation.tex
\section{Discussion and Conclusions}% and Societal Impact}
\label{section:limitation}
%There are several directions to which this work can be extended.
This work has some limitations and can be extended in several directions. 
%For instance, 
% First, Pearl's FD criterion is not complete in the sense that there might exist a causal effect that can be computed by the FD adjustment formula (Eq.~(\ref{eq:fd})) but the set \(\*Z\) does not satisfy the FD criterion.
First, Pearl's FD criterion is not complete with respect to the FD adjustment formula (Eq.~(\ref{eq:fd})).
While the BD criterion has been generalized to a complete criterion for BD adjustment \citep{shpitser:etal10}, it is an interesting open problem to come up with a complete criterion for sets satisfying the FD adjustment. 
Second, this work assumes that the causal diagram is given (or inferred based on scientists' domain knowledge and/or data).
Although this assumption is quite common throughout the causal inference literature, more recent work has moved to finding BD admissible sets given incomplete or partially specified causal diagrams, e.g., maximal ancestral graphs (MAGs) \citep{zander:etal14}, partial ancestral graphs (PAGs) \citep{perkovic:etal18}, and completed partially directed acyclic graphs (CPDAGs) \citep{perkovic:etal18}.
There are algorithms capable of performing causal effect identification in a data-driven fashion from an equivalence class \citep{jaber:etal18-id-markov-eq,jaber:19a,jaber:19b,jaber:19c}.
It is an interesting and certainly challenging future work to develop algorithms for finding FD admissible sets in these types of graphs.
% \textcolor{red}{Added references on data-driven methods for finding and listing BD admissible sets, and testing FD-admissibility.}
Some recent work has proposed data-driven methods for finding and listing BD admissible sets, using an anchor variable, when the underlying causal diagram is unknown \cite{entner:etal13,cheng:22,shah:22}.
A criterion for testing FD-admissibility of a given set using data and an anchor variable is also available \cite{bhattacharya1:22}.
Other interesting future research topics include developing algorithms for finding minimal, minimum, and minimum cost FD adjustment sets, which are available for the BD adjustment sets \citep{zander:etal19}, as well as algorithms for finding conditional FD adjustment sets \cite{hun:bar2019,fulcher2019robust}.
Having said all of that, we believe that the results developed in this paper is a necessary step towards solving these more challenging problems.

%% file: 7.conclusion.tex
After all, we started from the observation that 
identification is not restricted to BD adjustment, and 
Pearl's FD criterion provides a classic strategy for estimating causal effects from observational data and qualitative knowledge encoded in the form of a causal diagram. The criterion is drawing more attention in recent years and   statistically efficient and doubly robust
estimators have been developed for estimating the FD estimand from finite samples. 
In this paper, we develop algorithms that given a causal diagram $\G$, find an admissible FD set (Alg.~\ref{alg:findfdset} \textsc{FindFDSet}, Thm.~\ref{thm:findfdset}) and enumerate all admissible FD sets with polynomial delay (Alg.~\ref{alg:listfdsets} \textsc{ListFDSets}, Thm.~\ref{thm:listfdsets}). 
%The setting is challenging and involve an interesting interplay  understanding of causality and complexity analysis, including of this 
%We studied the identification of causal effects by front-door adjustment and the problem of enumerating all valid front-door adjustment sets. We developed the function \textsc{FindFDSet} that outputs a front-door adjustment set, and the algorithm \textsc{ListFDSets} that lists all front-door adjustment sets with polynomial delay.
We hope that the methods and algorithms proposed in this work will help scientists to use the FD strategy for causal effects estimation in the practical applications and are useful for scientists in study design to select covariates based on desired properties, including cost, feasibility, and statistical power.  
% Further, we believe the results of this paper is a necessary step towards solving problems under more complex settings, such as finding FD adjustment sets given an incomplete or partially specified causal diagram.
%Further, we believe the results of this paper is a necessary step towards solving problems under more complex settings, such as finding minimal/minimum or conditional FD adjustment sets, and finding FD adjustment sets given an incomplete or partially specified causal diagram.

%% file: 8.appendix.tex
\theoremstyle{plain}
\newtheorem{adxtheorem}{Theorem}
\newtheorem{adxlemma}{Lemma}
\newtheorem{adxprop}{Proposition}[section]
\newtheorem{adxcorollary}{Corollary}
\theoremstyle{definition}
\newtheorem{adxdefinition}{Definition}
\newtheorem{adxassumption}{Assumption}

\section{Appendix}

\begin{adxlemma}[Correctness of \textsc{GetCand2ndFDC}]
\label{adxlemma:getcand2ndfdc}
%Let \(\G\) be a causal diagram, \(\*X,\*Y\) disjoint sets of variables, and \(\*I, \*R\) sets of variables where \(\*I \subseteq \*R\).
\textsc{GetCand2ndFDC}\((\G, \*X, \*I, \*R)\) generates a set of variables \(\*R'\) with \(\*I \subseteq \*R' \subseteq \*R\) such that \(\*R'\) consists of all and only variables \(v\) that satisfies the second condition of the FD criterion relative to \((\*X,\*Y)\).
Further, every subset \(\*Z \subseteq \*R'\) satisfies the second condition of the FD criterion relative to \((\*X,\*Y)\), and every set \(\*Z\) with \(\*I \subseteq \*Z \subseteq \*R\) that satisfies the second condition of the FD  criterion relative to \((\*X,\*Y)\) must be a subset of \(\*R'\).
\end{adxlemma}

\begin{proof}
\textsc{GetCand2ndFDC} iterates through every node \(v \in \*R\).
% For each \(v\), the function \(\textsc{TestSep}(\G_{\underline{\*X}}, \*X, v, \emptyset)\) is called in line~\ref{func:getcand2ndfdc:test} to check if \(\emptyset\) is a separator of \(\*X\) and \(v\) in \(\G_{\underline{\*X}}\), i.e., whether there exists an open BD path from \(\*X\) to \(v\) or not.
For each \(v\), the function \(\textsc{TestSep}(\G_{\underline{\*X}}, \*X, v, \emptyset)\) is called in line~5 to check if \(\emptyset\) is a separator of \(\*X\) and \(v\) in \(\G_{\underline{\*X}}\), i.e., whether there exists an open BD path from \(\*X\) to \(v\) or not.
If \textsc{TestSep} returns True, then there is no open BD path from \(\*X\) to \(v\) and \(v\) satisfies the second condition of the FD criterion relative to \((\*X,\*Y)\).
In this case, \(v\) is kept in \(\*R'\).
Otherwise, if \textsc{TestSep} returns False, then there exists an open BD path from \(\*X\) to \(v\).
% By Lemma~\ref{lemma:findfdset:backdoor}, for every set \(\*Z\) that includes \(v\), there exists an open BD path from \(\*X\) to \(\*Z\).
By definition, for every set \(\*Z\) that includes \(v\), there exists an open BD path from \(\*X\) to \(\*Z\).
\(\*Z\) violates the second condition of the FD criterion relative to \((\*X,\*Y)\), and thus \(v\) is removed from \(\*R'\).
A special case is when \(v \in \*I\).
\textsc{GetCand2ndFDC} returns \(\perp\) because \(\*R'\) will not include any subset \(\*Z\) with \(\*I \subseteq \*Z\) that satisfies the second condition of the FD criterion relative to \((\*X,\*Y)\).

At the end of the function, \textsc{GetCand2ndFDC} has generated a set \(\*R'\) that includes all and only variables \(v\) that satisfies the second condition of the FD criterion relative to \((\*X,\*Y)\).
% By Lemma~\ref{lemma:findfdset:backdoor}, there exists no BD path from \(\*X\) to \(\*Z\) if and only if there exists no BD path from \(\*X\) to every \(v \in \*Z\).
By definition, there exists no BD path from \(\*X\) to \(\*Z\) if and only if there exists no BD path from every \(x \in \*X\) to every \(v \in \*Z\).
Hence, every subset \(\*Z \subseteq \*R'\) satisfies the second condition of the FD criterion relative to \((\*X,\*Y)\), and \(\*R'\) contains all and only sets \(\*Z\) with \(\*I \subseteq \*Z \subseteq \*R\) that satisfies the second condition of the FD criterion relative to \((\*X,\*Y)\).
\end{proof}

% Prop - GetCand2ndFDC - Complexity
\begin{adxprop}[Complexity of \textsc{GetCand2ndFDC}]
\label{adxprop:getcand2ndfdc:complexity}
\textsc{GetCand2ndFDC} runs in \(O(n(n+m))\) time where \(n\) and \(m\) represent the number of nodes and edges in \(\G\).
\end{adxprop}

\begin{proof}
\textsc{GetCand2ndFDC} iterates through all variables in \(\*R\) of size \(O(n)\).
For each variable \(v \in \*R\), the function \textsc{TestSep} is called, which takes \(O(n+m)\) time \citep{zander:etal19}.
\end{proof}

% Prop - Correctness of GetNeighbors
\begin{adxprop}[Correctness of \textsc{GetNeighbors}]
\label{adxprop:getneighbors}
Let \(\G\) be an undirected graph and \(v\) a variable in \(\G\).
\textsc{GetNeighbors} correctly outputs all observed neighbors \(\*N\) of \(v\) in \(\G\).
\textsc{GetNeighbors} runs in \(O(n+m)\) time where \(n\) and \(m\) represent the number of nodes and edges in \(\G\).
\end{adxprop}

\begin{proof}
\textsc{GetNeighbors} computes \(\*N\), all adjacent nodes of \(v\) in \(\G\) that are observed.
Also, all latent adjacent nodes \(\*L\) of \(v\) need to be considered because there might exist some observed adjacent nodes \(\*O\) of \(\*L\) where \(\*O\) belongs to observed neighbors of \(v\).
If \(\*L\) is empty, then all adjacent nodes of \(v\) are observed, and thus \textsc{GetNeighbors} returns \(\*N\).
Otherwise, \textsc{GetNeighbors} performs BFS from \(\*L\), searching for all observed neighbors of \(\*L\).
The nodes \(v\), \(\*N\) and \(\*L\) are marked as visited to guarantee that the nodes will not be visited more than once.

When BFS is performed, one latent node \(u\) is popped from \(\*Q\) at a time.
Then, all observed adjacent nodes \(\*O\) of \(u\) (that have not been visited before) are computed and added to \(\*N\).
Further, there may exist some latent adjacent nodes \(\*L'\) of \(u\) that have not been visited, and then there may exist some observed neighbors of \(\*L'\) as well.
Hence, \(\*L'\) is inserted into \(\*Q\) and all nodes in \(\*L'\) is marked as visited.
The procedure continues until \(\*Q\) becomes empty.

At the end of while loop, \(\*N\) must include all and only observed neighbors of \(v\) in \(\G\) because all observed adjacent nodes of \(v\) are added to \(\*N\), and for all latent adjacent nodes \(\*L\) of \(v\), all observed neighbors of \(\*L\) are also added to \(\*N\).

\textsc{GetNeighbors} runs in \(O(n+m)\) time because, while performing BFS, every node and edge in \(\G\) will be visited at most once.
\end{proof}

% Function - GetNeighbors
\begin{figure}

\begin{algorithmic} [1]
\Function {GetNeighbors} {$v, \G$}
% \Function {\textsc{GetNeighbors}(\(v, \G\))}

    % \State {\bfseries Input:} \(\G\) a causal diagram; \(\*X,\*Y\) disjoint sets of variables; \(\*R''\) sets of variables.
    
    \State {\bfseries Output:} \(\*N\) all observed neighbors of \(v\) in an undirected graph \(\G\).
    
    \State \(\*N \gets \) observed adjacent nodes of \(v\) in \(\G\), mark \(v\) and all \(w \in \*N\) as visited
    \State \(\*L \gets \) latent adjacent nodes of \(v\) in \(\G\), mark all \(w \in \*L\) as visited
    \State \(\*Q \gets \*L\)
    
    \While {\(\*Q \neq \emptyset\)}
        \State \(u \gets \*Q.\textsc{pop}()\)
        
        \State \(\*O \gets \) observed adjacent nodes of \(u\) in \(\G\) that have not been visited
        \State \(\*N \gets \*N \cup \*O\), mark all \(w \in \*O\) as visited
        \State \(\*L \gets \) latent adjacent nodes of \(u\) in \(\G\) that have not been visited
        \State \(\*Q.\textsc{insert}(\*L)\), mark all \(w \in \*L\) as visited
    \EndWhile
    
    \State \textbf{return} \(\*N\)
\EndFunction
\end{algorithmic}
\caption{A function that outputs all observed neighbors of a given variable.}
\label{func:getneighbors}
\end{figure}

% Prop - Correctness of GetDep
\begin{adxprop}[Correctness of \textsc{GetDep}]
\label{adxprop:getdep}
Let \(\G\) be a causal graph, \(\*X,\*Y,\*R'\) disjoint sets of variables, and \(\*T\) a set of variables where \(\*T \subseteq \*R'\).
If there exists a set of variables \(\*Z' \subseteq \*R' \setminus \*T\) such that \(\*T \cup \*Z'\) satisfies the third condition of the FD criterion relative to \((\*X,\*Y)\), \textsc{GetDep} outputs \(\*Z'\), or outputs \(\perp\) if none exists, in \(O(n^2 (n+m))\) time where \(n\) and \(m\) represent the number of nodes and edges in \(\G\).
\end{adxprop}

% If any node \(v \in \De{\*Y}\) is in \(\*Z\), then there exists a BD path from \(\*Z\) to \(\*Y\) due to the BD path \(\pi\) from \(v\) to \(y \in \*Y\) (i.e., a causal path from \(y\) to \(v\)).
% Note that \(\*X\) cannot be blocking \(\pi\); otherwise, there will be a causal path from \(y\) to some \(x \in \*X\), which violates the assumption that \(\*X\) causes \(\*Y\) and not vice versa.

\begin{proof}
% \textsc{GetDep} prohibits all descendants of \(\*Y\) to be present in \(\*T\) and \(\*R'\).
% That is because, with \(\*Z' \subseteq \*R' \setminus \*T\) and \(\*Z = \*T \cup \*Z'\), any \(\*Z\) will violate the third condition of the FD criterion relative to \((\*X,\*Y)\).
% If any node \(v \in \De{\*Y}\) is in \(\*Z\), then there exists a BD path from \(\*Z\) to \(\*Y\) due to the BD path \(\pi\) from \(v\) to \(y \in \*Y\).
% \(\*X\) cannot be blocking \(\pi\); otherwise, there will be a causal path from \(y\) to some \(x \in \*X\), which violates the assumption that \(\*X\) causes \(\*Y\) and not vice versa.
% Hence, \(\*Z\) violates the third condition of the FD criterion relative to \((\*X,\*Y)\).
% \textsc{GetDep} removes all descendants of \(\*Y\) from \(\*R'\) as no nodes in \(\De{\*Y}\) can be in any subset \(\*Z'\) such that \(\*Z = \*T \cup \*Z'\) satisfies the third condition of the FD criterion relative to \((\*X,\*Y)\).
% If any node \(v \in \De{\*Y}\) is in \(\*Z'\), then there exists a BD path from \(\*Z\) to \(\*Y\) due to the BD path \(\pi\) from \(v\) to \(y \in \*Y\).
% \(\*X\) cannot be blocking \(\pi\); otherwise, there will be a causal path from \(y\) to some \(x \in \*X\), which violates the assumption that \(\*X\) causes \(\*Y\) and not vice versa.
% \(\*Z\) violates the third condition of the FD criterion relative to \((\*X,\*Y)\), and thus \(\De{\*Y}\) is removed from \(\*R'\).

\textsc{GetDep} constructs the graph \(\G'\) by starting from the subgraph over \(\An{\*T \cup \*X \cup \*Y}\), and then converting all bidirected edges \(A \leftrightarrow B\) into a single latent node \(L_{AB}\) and two edges \(L_{AB} \rightarrow A\) and \(L_{AB} \rightarrow B\).
All outgoing edges of \(\*T\) are removed from \(\G'\) to create \(\G''\), which is then moralized to construct an undirected graph \(\mathcal{M}\).
After, \(\*X\) is removed from \(\mathcal{M}\).
The construction of \(\mathcal{M}\) is based on the property that \(\*T\) and \(\*Y\) are $d$-separated by \(\*X\) in \(\G\) if and only if \(\*X\) is a \(\*T\) - \(\*Y\) node cut (i.e., removing \(\*X\) disconnects \(\*T\) from \(\*Y\)) in \(\G_0 = \textsc{moralize}(\G_{ \An{\*T \cup \*X \cup \*Y} })\) \citep{lauritzen:96}.
The two tweaks: 1) removing all outgoing edges of \(\*T\) from \(\G'\) before moralization, and 2) removing \(\*X\) from \(\mathcal{M}\) after moralization are added to ensure that all paths from \(\*T\) to \(\*Y\) in \(\mathcal{M}\) are the BD paths from \(\*T\) to \(\*Y\) that cannot be blocked by \(\*X\) in \(\G\).

\textsc{GetDep} performs BFS from \(\*T\) to \(\*Y\) in \(\mathcal{M}\).
Whenever a node \(u\) is visited, \textsc{GetDep} obtains the non-visited, observed neighbors \(\*{NR}\) of \(u\) in \(\mathcal{M}\) that belong to \(\*R'\).
% All observed neighbors of \(u\) in \(\mathcal{M}\) are obtained by calling the function \textsc{GetNeighbors}(\(u, \mathcal{M}\)) (by Prop.~\ref{adxprop:getneighbors}).
All observed neighbors of \(u\) in \(\mathcal{M}\) are obtained by calling the function \textsc{GetNeighbors}(\(u, \mathcal{M}\)) (by Prop.~A.2).
Then, \(\mathcal{M}\) is reconstructed by moralizing the graph \(\G'' = \G'_{\underline{\*T \cup \*Z' \cup \*{NR}}}\) that removes all outgoing edges of \(\*T \cup \*Z' \cup \*{NR}\) from \(\G'\), and then removing \(\*X\) from \(\mathcal{M}\) after.
All outgoing edges of \(\*{NR}\) are removed (in addition to those of \(\*T \cup \*Z'\)) to check if removing all outgoing edges of \(\*{NR}\) contributes to disconnecting BD paths from \(\*T\) to \(\*Y\) that cannot be blocked by \(\*X\) in \(\G\).
In other words, \(\*{NR}\) may belong to \(\*Z'\) such that \(\*Z\) satisfies the third condition of the FD criterion relative to \((\*X,\*Y)\).
Hence, \(\*{NR}\) is added to \(\*Z'\).

However, there might exist some BD path \(\pi\) from \(w \in \*{NR}\) to \(y \in \*Y\) that cannot be blocked by \(\*X\) in \(\G\).
If \(\pi\) cannot be disconnected from \(w\) to \(y\), then \(\*Z\) will violate the third condition of the FD criterion relative to \((\*X,\*Y)\).
We need to check if there exists such \(\pi\).
\textsc{GetDep} constructs a set \(\*{NR'}\), a set of all variables \(w\) in \(\*{NR}\) such that there exists an incoming arrow into \(w\) in \(\G\).
Also, there might exist some observed neighbors \(\*N'\) of \(u\) in \(\mathcal{M}\) that are still reachable from \(u\), even after removing all outgoing edges of \(\*T \cup \*Z' \cup \*{NR}\) (which is reflected by the construction of \(\mathcal{M}\)).
Hence, the union \(\*N\) of two sets, \(\*N'\) and \(\*{NR'}\), are inserted into \(\*Q\) to check if any node in \(\*N\) is reachable to \(\*Y\).

The BFS continues until either a node \(y \in \*Y\) is visited, or no more node can be visited.
We explain further by each case.

\begin{enumerate}
    % \item There exists no set \(\*Z'\) such that \(\*Z\) satisfies the third condition of the FD criterion relative to \((\*X,\*Y)\).
    % That is because, while performing BFS from \(\*T\) to \(\*Y\) in \(\mathcal{M}\), whenever a node \(u\) is being visited, \textsc{GetDep} removes all outgoing edges of \(\*{NR}\), all non-visited neighbors of \(u\) that belong to \(\*R'\).
    % Consider the case when a node in \(\*Y\) is visited, and let \(\pi\) be a BD path from \(\*T\) to \(\*Y\) in \(\mathcal{M}\) where all nodes in \(\pi\) have been visited.
    % Since all nodes in \(\pi\) are visited, for all variables \(w \in \*R'\) that intersect \(\pi\), all outgoing edges of \(w\) must have been removed in \(\G^{L'}\) and \(\mathcal{M}\) was constructed based on \(\G^{L'}\) without any outgoing edges of \(w\).
    % However, a node in \(\*Y\) was still reached, which implies that removing all outgoing edges of \(w\) did not disconnect \(\pi\) from \(\*T\) to \(\*Y\).
    % Removing outgoing edges of any combinations of nodes in \(\*R'\) will not disconnect \(\pi\) from \(\*T\) to \(\*Y\) either.
    % Thus, there exists no subset \(\*Z'\) such that all BD paths from \(\*Z\) to \(\*Y\) are blocked by \(\*X\).
    % \textsc{GetDep} returns \(\perp\).
    \item A node \(y \in \*Y\) is visited.
    There exists no set \(\*Z'\) such that \(\*Z = \*T \cup \*Z'\) satisfies the third condition of the FD criterion relative to \((\*X,\*Y)\).
    % That is because, while performing BFS from \(\*T\) to \(\*Y\) in \(\mathcal{M}\), whenever a node \(u\) is being visited, \textsc{GetDep} removed all outgoing edges of \(\*T \cup \*Z' \cup \*{NR}\).
    Let \(\pi\) be a BD path from \(t \in \*T\) to \(y\) in \(\mathcal{M}\) where all nodes in \(\pi\) are visited by performing BFS from \(t\) to \(y\).
    Since all nodes in \(\pi\) are visited, for all variables \(w \in \*R'\) that intersect \(\pi\), all outgoing edges of \(w\) must have been removed in \(\G''\) and \(\mathcal{M}\) was constructed based on \(\G''\).
    However, \(y\) was still reached, which implies that removing all outgoing edges of \(w\) did not disconnect \(\pi\) from \(t\) to \(y\).
    Removing all outgoing edges of \(\*R'\) will not disconnect \(\pi\) from \(t\) to \(y\) either.
    Thus, there exists no set \(\*Z'\) such that all BD paths from \(\*Z\) to \(\*Y\) are blocked by \(\*X\) in \(\G\).
    \textsc{GetDep} returns \(\perp\).
    
    \item No more node is left to be visited.
    All BD paths from \(\*T\) to \(\*Y\) that cannot be blocked by \(\*X\) in \(\G\) have been disconnected by removing all outgoing edges of \(\*Z\) while ensuring that there exists no BD path from \(\*Z\) to \(\*Y\) that cannot be blocked by \(\*X\) in \(\G\).
    All BD paths from \(\*Z\) to \(\*Y\) are blocked by \(\*X\), and thus \(\*Z\) satisfies the third condition of the FD criterion relative to \((\*X,\*Y)\).
    \textsc{GetDep} returns the set \(\*Z'\).
\end{enumerate}

For the time complexity, \textsc{moralize} runs in \(O(n^2)\) time.
\textsc{moralize} checks over every pair of nodes (of size \(O(n^2)\)) and adds an undirected edge between each non-adjacent pair if both nodes share a common child.
Then, \textsc{moralize} converts all directed edges into undirected edges where the number of edges may be of \(O(n^2)\) in the worst case scenario.
The BFS takes \(O(n^2 (n+m))\) time in total since all nodes and edges may be visited at most once (i.e., \(O(n+m)\) entities) where visiting a single node takes \(O(n^2)\) time where the dominating factor is the runtime of \textsc{moralize}.
% By Prop.~\ref{adxprop:getneighbors}, \textsc{GetNeighbors} runs in \(O(n+m)\) time.
% Hence, \textsc{GetDep} runs in \(O(n^2 (n+m))\) time.
By Prop.~A.2, \textsc{GetNeighbors} runs in \(O(n+m)\) time.
Hence, \textsc{GetDep} runs in \(O(n^2 (n+m))\) time.

% For the runtime, \textsc{moralize} runs in \(O(n^2)\) time.
% \textsc{moralize} checks over every pair of nodes (of size \(O(n^2)\)) and adds an undirected edge between each pair if both nodes share a common child.
% Then, \textsc{moralize} converts all directed edges into undirected edges where the number of edges may be of \(O(n^2)\) in the worst case scenario.
% The BFS takes takes \(O((n+m)^2)\) time in total since all nodes and edges may be visited at most once (i.e., \(O(n+m)\) entities) where visiting a single node takes \(O(n+m)\) time in total.
% Constructing each of the sets, \(\*{NR}\), \(\*{NR}'\), \(\*N\), and \(\*N'\) takes at most \(O(n+m)\) time.
% The function \textsc{GetNeighbors} runs in \(O(n+m)\) time.
\end{proof}

% Lemma - GetCand3rdFDC
% \begin{adxlemma}{\ref{lemma:findfdset:3rd}}[Correctness of \textsc{GetCand3rdFDC}]
\begin{adxlemma}[Correctness of \textsc{GetCand3rdFDC}]
\label{adxlemma:getcand3rdfdc}
% \textsc{GetCand3rdFDC}\((\G, \*X, \*Y, \*I, \*R')\) in Step 2 of Alg.~\ref{alg:findfdset} generates a set of variables \(\*R''\) where \(\*I \subseteq \*R'' \subseteq \*R'\).
\textsc{GetCand3rdFDC}\((\G, \*X, \*Y, \*I, \*R')\) in Step 2 of Alg.~1 generates a set of variables \(\*R''\) where \(\*I \subseteq \*R'' \subseteq \*R'\).
\(\*R''\) consists of all and only variables \(v\) such that there exists a subset \(\*Z\) with \(\*I \subseteq \*Z \subseteq \*R'\) and \(v \in \*Z\) that satisfies the third condition of the FD criterion relative to \((\*X,\*Y)\).
Further, every set \(\*Z\) with \(\*I \subseteq \*Z \subseteq \*R\) that satisfies both the second and the third conditions of the FD criterion must be a subset of \(\*R''\).
\end{adxlemma}

\begin{proof}
The proof consists of two parts.

\begin{enumerate}
    \item \(\*R''\) consists of all and only variables \(v\) such that there exists a subset \(\*Z\) with \(\*I \subseteq \*Z \subseteq \*R'\) and \(v \in \*Z\) that satisfies the third condition of the FD criterion relative to \((\*X,\*Y)\).
    
    \textsc{GetCand3rdFDC} iterates through all variables \(v\) in \(\*R'\).
    % By Lemma~\ref{adxlemma:getcand2ndfdc}, every set \(\*Z\) with \(\*I \subseteq \*Z \subseteq \*R\) that satisfies the second condition of the FD criterion relative to \((\*X,\*Y)\) must be a subset of \(\*R'\).
    By Lemma~1, every set \(\*Z\) with \(\*I \subseteq \*Z \subseteq \*R\) that satisfies the second condition of the FD criterion relative to \((\*X,\*Y)\) must be a subset of \(\*R'\).
    % without while
    % For each \(v\), if \textsc{GetDep} returns \(\perp\), then for every set \(\*Z\) with \(\*I \subseteq \*Z \subseteq \*R'\) and \(v \in \*Z\), \(\*Z\) violates the third condition of the FD criterion relative to \((\*X,\*Y)\) (by Prop.~\ref{adxprop:getdep}).
    For each \(v\), if \textsc{GetDep} returns \(\perp\), then for every set \(\*Z\) with \(\*Z \subseteq \*R'\) and \(v \in \*Z\), \(\*Z\) violates the third condition of the FD criterion relative to \((\*X,\*Y)\) (by Prop.~A.3).
    Hence, \(v\) is removed from \(\*R''\).
    All such \(v\)'s (i.e., \(v\) such that \textsc{GetDep} had returned \(\perp\)) will be removed from \(\*R''\).
    If \(v \in \*I\), then \textsc{GetCand3rdFDC} returns \(\perp\) as no \(\*Z\) with \(\*I \subseteq \*Z \subseteq \*R'\) and \(v \in \*Z\) will satisfy the third condition of the FD criterion relative to \((\*X,\*Y)\).
    At the end of for loop, we have that \(\*R''\) consists all and only variables \(v\) such that there exists a subset \(\*Z\) with \(\*I \subseteq \*Z \subseteq \*R'\) and \(v \in \*Z\) that satisfies the third condition of the FD criterion relative to \((\*X,\*Y)\).

    \item Every set \(\*Z\) with \(\*I \subseteq \*Z \subseteq \*R\) that satisfies both the second and the third conditions of the FD criterion relative to \((\*X,\*Y)\) must be a subset of \(\*R''\).
    
    % By Lemma~\ref{adxlemma:getcand2ndfdc}, every set \(\*Z\) with \(\*I \subseteq \*Z \subseteq \*R\) that satisfies the second condition of the FD criterion relative to \((\*X,\*Y)\) must be a subset of \(\*R'\).
    By Lemma~1, every set \(\*Z\) with \(\*I \subseteq \*Z \subseteq \*R\) that satisfies the second condition of the FD criterion relative to \((\*X,\*Y)\) must be a subset of \(\*R'\).
    We restrict the scope of \(\*Z\) into \(\*I \subseteq \*Z \subseteq \*R'\) and show that every \(\*Z\) that satisfies the third condition of the FD criterion relative to \((\*X,\*Y)\) must be a subset of \(\*R''\).
    
    When \textsc{GetCand3rdFDC} iterates through all variables in \(\*R'\), every \(u \in \*Z\) must have been checked since \(\*Z \subseteq \*R'\).
    % For each \(u \in \*Z\), \textsc{GetDep} must have returned a set of variables since there exists a subset \(\*Z' = \*Z \setminus \{u\} \subseteq \*R' \setminus \{u\}\) such that \(\*Z\) satisfies the third condition of the FD criterion relative to \((\*X,\*Y)\) (by Prop~\ref{adxprop:getdep}).
    For each \(u \in \*Z\), \textsc{GetDep} must have returned a set of variables since there exists a subset \(\*Z' = \*Z \setminus \{u\} \subseteq \*R' \setminus \{u\}\) such that \(\*Z\) satisfies the third condition of the FD criterion relative to \((\*X,\*Y)\) (by Prop~A.3).
    \textsc{GetCand3rdFDC} removes all and only variables \(v\) from \(\*R'\) such that there exists no set \(\*Z'\) with \(\*I \subseteq \*Z' \subseteq \*R'\) and \(v \in \*Z'\) that satisfies the third condition of the FD criterion relative to \((\*X,\*Y)\).
    If \(\*Z\) includes any such \(v\), then it is a contradiction as \(\*Z\) will violate the third condition of the FD criterion relative to \((\*X,\*Y)\).
    Hence, \(\*Z\) must be a subset of \(\*R''\).
    
    % Since all BD paths from \(\*Z\) to \(\*Y\) are blocked by \(\*X\), there exists no BD path from \(\*Z\) to \(\*Y\) that cannot be blocked by \(\*X\).
    % Then, in \(\G_{\underline{\*Z}}\), for all \(v \in \*Z\), all BD paths from \(v\) to \(\*Y\) that cannot be blocked by \(\*X\) must have been disconnected.
    % Otherwise, there exists some \(v \in \*Z\) with an open BD path from \(v\) to \(\*Y\) that cannot be blocked by \(\*X\), even when all outgoing edges of \(\*Z\) were removed.
    % In other words, \(\*Z\) violates the third condition of FD criterion relative to \((\*X,\*Y)\).
    % We have that, for every \(v \in \*Z\), there exists a subset \(\*Z' = \*Z \setminus \{v\} \subseteq \*R' \setminus \{v\}\) such that all BD paths from the union of sets \(\*Z = \{v\} \cup \*Z'\) to \(\*Y\) are blocked by \(\*X\).
    
    % Note that \(\*Z' \subseteq \*R'' \setminus \{v\}\) holds.
    % By the construction of R'',
    
    % When \textsc{GetCand3rdFDC} iterates through all variables \(v \in \*R'\), every \(u \in \*Z\) must have been checked since \(\*Z \subseteq \*R'\).
    % For each \(u \in \*Z\), \textsc{GetDep} must have returned a set of variables since there exists a subset \(\*Z' = \*Z \setminus \{u\} \subseteq \*R' \setminus \{u\}\) such that all BD paths from the union of sets \(\*Z = \{u\} \cup \*Z'\) to \(\*Y\) are blocked by \(\*X\) (from Prop. \ref{adxprop:getdep}).
    % Then, every \(u \in \*Z\) must have been kept in \(\*R''\) and thus \(\*Z\) must belong to \(\*R''\).
\end{enumerate}
\end{proof}

% % Lemma - GetCand3rdFDC - Complexity
\begin{adxprop}[Complexity of \textsc{GetCand3rdFDC}]
\label{adxprop:getcand3rdfdc:complexity}
\textsc{GetCand3rdFDC} runs in \(O(n^3 (n+m))\) time where \(n\) and \(m\) represent the number of nodes and edges in \(\G\).
\end{adxprop}

\begin{proof}
\textsc{GetCand3rdFDC} iterates through all variables \(v\) in \(\*R'\) of size \(O(n)\).
The function \textsc{GetDep} will be called once per loop.
% In the worst case scenario, a single variable is removed from \(\*R''\) on each while loop, having \(O(n)\) while loops in total.
% Hence, \textsc{GetDep} may be called \(O(n^2)\) times.
% By Prop.~\ref{adxprop:getdep}, \textsc{GetDep} runs in \(O(n^2 (n+m))\) time.
By Prop.~A.3, \textsc{GetDep} runs in \(O(n^2 (n+m))\) time.
In total, the running time of \textsc{GetCand3rdFDC} is \(O(n^3 (n+m))\).
\end{proof}

% Lemma - R''
% \begin{adxlemma}{\ref{lemma:findfdset:r}}
\begin{adxlemma}
\label{adxlemma:findfdset:r}
% \(\*R''\) generated by \textsc{GetCand3rdFDC} (in Step 2 of Alg.~\ref{alg:findfdset}) satisfies the third condition of the FD criterion, that is, all BD paths from \(\*R''\) to \(\*Y\) are blocked by \(\*X\).
\(\*R''\) generated by \textsc{GetCand3rdFDC} (in Step 2 of Alg.~1) satisfies the third condition of the FD criterion, that is, all BD paths from \(\*R''\) to \(\*Y\) are blocked by \(\*X\).
\end{adxlemma}

\begin{proof}
% By Lemma~\ref{adxlemma:getcand3rdfdc}, for every variable \(v \in \*R''\), there exists a subset \(\*Z' \subseteq \*R' \setminus \{v\}\) such that \(\*Z = \{v\} \cup \*Z'\) satisfies the third condition of the FD criterion relative to \((\*X,\*Y)\).
By Lemma~2, for every variable \(v \in \*R''\), there exists a subset \(\*Z' \subseteq \*R' \setminus \{v\}\) such that \(\*Z = \{v\} \cup \*Z'\) satisfies the third condition of the FD criterion relative to \((\*X,\*Y)\).
In other words, there is no BD path from \(\*Z\) to \(\*Y\) that cannot be blocked by \(\*X\) in \(\G\).
All BD paths from \(v\) to \(\*Y\) that cannot be blocked by \(\*X\) are disconnected in \(\G_{\underline{\*Z}}\) by removing all outgoing edges of \(v\) and \(\*Z'\) in \(\G\).
% Consider the graph \(\G_{\underline{\*R''}}\) where all outgoing edges of \(\*Z\) as well as those of \(\*R'' \setminus \*Z\) are removed (\(\*Z \subseteq \*R''\) holds by Lemma~\ref{adxlemma:getcand3rdfdc}).
Consider the graph \(\G_{\underline{\*R''}}\) where all outgoing edges of \(\*Z\) as well as those of \(\*R'' \setminus \*Z\) are removed (\(\*Z \subseteq \*R''\) holds by Lemma~2).
% Since \(\*Z \subseteq \*R''\) holds (By Lemma~\ref{adxlemma:getcand3rdfdc}), in \(\G_{\underline{\*R''}}\), all outgoing edges of \(\*Z\) as well as those of \(\*R'' \setminus \*Z\) are removed.
Removing more outgoing edges (i.e., in \(\G_{\underline{\*R''}}\)) will not re-connect the BD paths that have already been disconnected in \(\G_{\underline{\*Z}}\).
Hence, all BD paths from \(v\) to \(\*Y\) that cannot be blocked by \(\*X\) will be disconnected in \(\G_{\underline{\*R''}}\).
Then, for every variable \(v \in \*R''\), all BD paths from \(v\) to \(\*Y\) that cannot be blocked by \(\*X\) will be disconnected in \(\G_{\underline{\*R''}}\).
All BD paths from \(\*R''\) to \(\*Y\) that cannot be blocked by \(\*X\) are disconnected in \(\G_{\underline{\*R''}}\).
All BD paths from \(\*R''\) to \(\*Y\) are blocked by \(\*X\) and thus \(\*R''\) satisfies the third condition of the FD criterion relative to \((\*X,\*Y)\).

% Consider an another subset \(\*Z'' = \*R'' \setminus \{v\}\).
% Then, we have \(\*Z' \subseteq \*Z''\).
% Since there is no BD path from \(\{v\} \cup \*Z'\) to \(\*Y\) that do not intersect \(\*X\), which implies that there is no BD path from \(\{v\} \cup \*Z''\) to \(\*Y\) that do not intersect \(\*X\).
% Consider two graphs \(\G_{\underline{\*Z'}}\) and \(\G_{\underline{\*Z''}}\).
% All BD paths from \(v\) to \(\*Y\) are disconnected by removing all outgoing edges of \(\*Z'\).
% Since \(\*Z' \subseteq \*Z''\) holds, in \(\G_{\underline{\*Z''}}\), we are removing all outgoing edges of \(\*Z'\) as well as the those of \(\*Z'' \setmius \*Z'\).
% Hence, all BD paths from \(v\) to \(\*Y\) will be disconnected by removing all outgoing edges of \(\*Z''\).
% take 2
% Consider the remaining set of variables \(\*Z'' \setminus \*Z'\).
% For each variable \(w\) in \(\*Z'' \setminus \*Z'\), there exists a subset \(\*Z_w \subseteq \*R'' \setminus \{w\}\) such that the union of sets \(\{w\} \cup \*Z_w\) satisfies the third condition of FD criterion relative to \((\*X,\*Y)\).
% Since \(\*Z_w \subseteq \*R''\) for all \(w\),
\end{proof}

% Function - GetCausalPathGraph
\begin{figure}

\begin{algorithmic} [1]
\Function {GetCausalPathGraph} {$\G, \*X, \*Y$}
% \Function {\textsc{GetCausalPathGraph}(\(\G, \*X, \*Y\))}

    \State {\bfseries Output:} \(\G'\) a causal path graph relative to \((\G, \*X, \*Y)\).
    
    \State \(\G'' \gets \G_{\*X \cup \*Y \cup PCP(\*X,\*Y)}\)
    \State \(\G' \gets \G''_{\overline{\*X}\underline{\*Y}}\)
    \State Remove all bidirected edges from \(\G'\)
    \State \textbf{return} \(\G'\)
\EndFunction
\end{algorithmic}
\caption{A function that constructs a causal path graph.} \label{func:getcausalpathgraph}

\end{figure}

\begin{adxprop}
\label{adxprop:cpg}
Let \(\G\) be a causal graph and \(\*X,\*Y\) disjoint sets of variables.
\textsc{GetCausalPathGraph} constructs a causal path graph \(\G'\) relative to \((\G,\*X,\*Y)\) in \(O(n+m)\) time where \(n\) and \(m\) represent the number of nodes and edges in \(\G\).
\end{adxprop}

\begin{proof}
% The construction of a causal path graph is immediate from Def.~\ref{definition:causalpathgraph}.
The construction of a causal path graph is immediate from Def.~2.
Constructing a subgraph \(\G_{\*X \cup \*Y \cup PCP(\*X,\*Y)}\), performing graph transformation \(\G''_{\overline{\*X}\underline{\*Y}}\), and removing all bidirected edges take \(O(n+m)\) time.
\end{proof}

% Def - Proper Causal Path
\begin{definition}{(Proper Causal Path \citep{shpitser:etal10})}
Let \(\*X,\*Y\) be set of nodes.
A causal path from a node in \(\*X\) to a node in \(\*Y\) is called proper if it does not intersect \(\*X\) except at the end point.
\end{definition}

% Lemma - FindFDSet - 1st
% \begin{adxlemma}{\ref{lemma:findfdset:separator}}
\begin{adxlemma}
\label{adxlemma:findfdset:separator}
Let \(\G\) be a causal graph and \(\*X,\*Y,\*Z\) disjoint sets of variables.
Let \(\G'\) be the causal path graph relative to \((\G,\*X,\*Y)\).
Then, \(\*Z\) satisfies the first condition of the FD criterion relative to \((\*X, \*Y)\) if and only if \(\*Z\) is a separator of \(\*X\) and \(\*Y\) in \(\G'\).
\end{adxlemma}

\begin{proof}
We prove the statement in both directions.
\begin{itemize}
    \item \textit{If case:}
    We show that \(\*Z\) satisfies the first condition of the FD criterion relative to \((\*X, \*Y)\).
    By the construction of \(\G'\), all paths from \(\*X\) to \(\*Y\) comprise of all and only proper causal paths from \(\*X\) to \(\*Y\).
    It is only necessary to check for all proper causal paths from \(\*X\) to \(\*Y\) since every non-proper causal path from \(\*X\) to \(\*Y\) must include a proper causal path from \(\*X\) to \(\*Y\) as a subpath.
    To witness, consider any non-proper causal path \(\pi = x_1 \rightarrow, \cdots, \rightarrow x_k \rightarrow, \cdots, \rightarrow y\) from a node \(x_1 \in \*X\) to a node \(y \in \*Y\).
    Since \(\pi\) is not proper, there must exist a node \(x_k \in \*X\) that intersects \(\pi\) at non-endpoint and there exists a subpath \(\pi' = x_k \rightarrow, \cdots, \rightarrow y\) such that \(\pi'\) is proper.
    Since \(\*Z\) is a separator of \(\*X\) and \(\*Y\) in \(\G'\), \(\*Z\) intercepts all causal paths from \(\*X\) to \(\*Y\) in \(\G\).
    
    \item \textit{Only if case:}
    We show that \(\*Z\) is a separator of \(\*X\) and \(\*Y\) in \(\G'\).
    By assumption, \(\*Z\) intercepts all causal paths from \(\*X\) to \(\*Y\) in \(\G\).
    By the construction of \(\G'\), all and only paths from \(\*X\) to \(\*Y\) must be causal.
    Thus, \(\*Z\) must be a separator of \(\*X\) and \(\*Y\) in \(\G'\).
\end{itemize}
\end{proof}

% Lemma - exists FD
% \begin{adxlemma}{\ref{lemma:findfdset:exists}}
\begin{adxlemma}
\label{adxlemma:findfdset:exists}
% There exists a set \(\*Z_0\) satisfying the FD criterion relative to \((\*X,\*Y)\) with \(\*I \subseteq \*Z_0 \subseteq \*R\) if and only if \(\*R''\) generated by \textsc{GetCand3rdFDC} (in Step 2 of Alg.~\ref{alg:findfdset}) satisfies the FD criterion relative to \((\*X,\*Y)\).
There exists a set \(\*Z_0\) satisfying the FD criterion relative to \((\*X,\*Y)\) with \(\*I \subseteq \*Z_0 \subseteq \*R\) if and only if \(\*R''\) generated by \textsc{GetCand3rdFDC} (in Step 2 of Alg.~1) satisfies the FD criterion relative to \((\*X,\*Y)\).
\end{adxlemma}

We prove the statement in both directions.

\begin{itemize}
    \item \textit{If case:} It is automatic with \(\*Z_0 = \*R''\).
    
    \item \textit{Only if case:}
    We prove the contrapositive of the statement: if \(\*R''\) is not a FD adjustment set relative to \((\*X,\*Y)\), then there does not exist any FD adjustment set \(\*Z_0\) relative to \((\*X,\*Y)\) with \(\*I \subseteq \*Z_0 \subseteq \*R\).
    On the following three items, we show that there does not exist any FD adjustment set \(\*Z_0\) relative to \((\*X,\*Y)\) with three disjoint intervals, \(\*I \subseteq \*Z_0 \subseteq \*R''\), \(\*R'' \subset \*Z_0 \subseteq \*R'\), and \(\*R' \subset \*Z_0 \subseteq \*R\), respectively.

    \begin{enumerate}
        \item Since \(\*R''\) is not a FD adjustment set relative to \((\*X,\*Y)\), \(\*R''\) must be violating the first condition of the FD criterion relative to \((\*X,\*Y)\).
        % That is because, by the construction of \(\*R''\), \(\*R''\) must satisfy the second condition of the FD criterion relative to \((\*X,\*Y)\) (by Lemma~\ref{adxlemma:getcand2ndfdc}) and the third condition of the FD criterion relative to \((\*X,\*Y)\) (by Lemma~\ref{adxlemma:findfdset:r}).
        That is because, by the construction of \(\*R''\), \(\*R''\) must satisfy the second condition of the FD criterion relative to \((\*X,\*Y)\) (by Lemma~1) and the third condition of the FD criterion relative to \((\*X,\*Y)\) (by Lemma~3).
        Then, \(\*R''\) does not intercept all causal paths from \(\*X\) to \(\*Y\).
        % By Lemma~\ref{adxlemma:findfdset:separator}, \(\*R''\) is not a FD adjustment set relative to \((\*X,\*Y)\) if and only if \(\*R''\) is not a separator of \(\*X\) and \(\*Y\) in \(\G'\), a causal path graph relative to \((\G,\*X,\*Y)\).
        No subset \(\*Z_0\) with \(\*I \subseteq \*Z_0 \subseteq \*R''\) will intercept all causal paths from \(\*X\) to \(\*Y\).
        \(\*Z_0\) violates the first condition of the FD criterion relative to \((\*X,\*Y)\)), and thus \(\*Z_0\) is not a FD adjustment set relative to \((\*X,\*Y)\).
        
        \item Consider a collection of sets \(\*Z_0\) with \(\*R'' \subset \*Z_0 \subseteq \*R'\).
        % By the construction of \(\*R''\) generated by \textsc{GetCand3rdFDC} (with \(\*R'' \subseteq \*R'\)), for all \(v \in \*R' \setminus \*R''\), there does not exist any set \(\*Z\) with \(\*I \subseteq \*Z \subseteq \*R'\) and \(v \in \*Z\) that satisfies the third condition of the FD criterion relative to \((\*X,\*Y)\) (by Lemma~\ref{adxlemma:getcand3rdfdc}).
        By the construction of \(\*R''\) generated by \textsc{GetCand3rdFDC} (with \(\*R'' \subseteq \*R'\)), for all \(v \in \*R' \setminus \*R''\), there does not exist any set \(\*Z\) with \(\*I \subseteq \*Z \subseteq \*R'\) and \(v \in \*Z\) that satisfies the third condition of the FD criterion relative to \((\*X,\*Y)\) (by Lemma~2).
        \(\*Z_0\) must include some \(v\), and thus \(\*Z_0\) violates the third condition of the FD criterion relative to \((\*X,\*Y)\).
        \(\*Z_0\) is not a FD adjustment set relative to \((\*X,\*Y)\).
        
        \item Consider a collection of sets \(\*Z_0\) with \(\*R' \subset \*Z_0 \subseteq \*R\).
        % By the construction of \(\*R'\) generated by \textsc{GetCand2ndFDC} (with \(\*R' \subseteq \*R\)), for all \(v \in \*R \setminus \*R'\), there exists an open BD path from \(\*X\) to \(v\) (By Lemma~\ref{adxlemma:getcand2ndfdc}).
        By the construction of \(\*R'\) generated by \textsc{GetCand2ndFDC} (with \(\*R' \subseteq \*R\)), for all \(v \in \*R \setminus \*R'\), there exists an open BD path from \(\*X\) to \(v\) (By Lemma~1).
        \(\*Z_0\) must be including some \(v\), and by definition, there exists an open BD path from \(\*X\) to \(\*Z_0\).
        \(\*Z_0\) violates the second condition of the FD criterion relative to \((\*X,\*Y)\) and \(\*Z_0\) is not a FD adjustment set relative to \((\*X,\*Y)\).
    \end{enumerate}

    Combining together the three items, we have that for all \(\*Z_0\) with \(\*I \subseteq \*Z_0 \subseteq \*R\), \(\*Z_0\) is not a FD adjustment set relative to \((\*X,\*Y)\).

\end{itemize}

% FindFDSet - Correctness
% \begin{adxtheorem}{\ref{thm:findfdset}}[Correctness of \textsc{FindFDSet}]
\begin{adxtheorem}[Correctness of \textsc{FindFDSet}]
\label{adxthm:findfdset}
Let \(\G\) be a causal graph, \(\*X,\*Y\) disjoint sets of variables, and \(\*I, \*R\) sets of variables such that \(\*I \subseteq \*R\).
Then, \textsc{FindFDSet}\((\G, \*X, \*Y, \*I, \*R)\) outputs a set \(\*Z\) with \(\*I \subseteq \*Z \subseteq \*R\) that satisfies the FD criterion relative to \((\*X,\*Y)\), or outputs \(\perp\) if none exists, in \(O(n^3 (n+m))\) time, where \(n\) and \(m\) represent the number of nodes and edges in \(\G\).
\end{adxtheorem}

\begin{proof}
% By Lemma~\ref{adxlemma:getcand3rdfdc}, every set \(\*Z\) with \(\*I \subseteq \*Z \subseteq \*R\) that satisfies both the second and the third conditions of the FD criterion relative to \((\*X,\*Y)\) must be a subset of \(\*R''\).
By Lemma~2, every set \(\*Z\) with \(\*I \subseteq \*Z \subseteq \*R\) that satisfies both the second and the third conditions of the FD criterion relative to \((\*X,\*Y)\) must be a subset of \(\*R''\).
% By Lemma~\ref{adxlemma:getcand2ndfdc}, \(\*R''\) satisfies the second condition of the FD criterion relative to \((\*X,\*Y)\).
By Lemma~1, \(\*R''\) satisfies the second condition of the FD criterion relative to \((\*X,\*Y)\).
% By Lemma~\ref{adxlemma:findfdset:r}, \(\*R''\) satisfies the third condition of the FD criterion relative to \((\*X,\*Y)\).
By Lemma~3, \(\*R''\) satisfies the third condition of the FD criterion relative to \((\*X,\*Y)\).
Let \(\*Z = \*R''\).
% Then, By Lemma~\ref{adxlemma:findfdset:separator}, \(\*Z\) is a FD adjustment set relative to \((\*X, \*Y)\) if and only if \(\*Z\) is a separator of \(\*X\) and \(\*Y\) in \(\G'\), a causal path graph relative to \((\G,\*X,\*Y)\).
Then, By Lemma~4, \(\*Z\) is a FD adjustment set relative to \((\*X, \*Y)\) if and only if \(\*Z\) is a separator of \(\*X\) and \(\*Y\) in \(\G'\), a causal path graph relative to \((\G,\*X,\*Y)\).
% \textsc{FindFDSet} outputs \(\*Z\) if and only if \(\*Z\) is a separator of \(\*X\) and \(\*Y\) in \(\G'\) (by calling \(\textsc{TestSep}(\G', \*X, \*Y, \*Z)\) at line~\ref{findfdset:testsep} and verifying \textsc{TestSep} is returning True).
\textsc{FindFDSet} outputs \(\*Z\) if and only if \(\*Z\) is a separator of \(\*X\) and \(\*Y\) in \(\G'\) (by calling \(\textsc{TestSep}(\G', \*X, \*Y, \*Z)\) at line~11 and verifying \textsc{TestSep} is returning True).
Hence, the outputted set \(\*Z\) is a FD adjustment set relative to \((\*X, \*Y)\) where \(\*I \subseteq \*Z \subseteq \*R\).
If \textsc{TestSep} returns False, then \(\*Z\) is not a FD adjustment set relative to \((\*X, \*Y)\) and \textsc{FindFDSet} outputs \(\perp\).
% By Lemma~\ref{adxlemma:findfdset:exists}, there does not exist any FD adjustment set \(\*Z_0\) relative to \((\*X,\*Y)\) with \(\*I \subseteq \*Z_0 \subseteq \*R\).
By Lemma~5, there does not exist any FD adjustment set \(\*Z_0\) relative to \((\*X,\*Y)\) with \(\*I \subseteq \*Z_0 \subseteq \*R\).

% For the running time, constructing \(\*R'\) takes \(O(n (n+m))\) time (by Prop.~\ref{adxprop:getcand2ndfdc:complexity}), and generating \(\*R''\) takes \(O(n^3  (n+m))\) time (by Prop.~\ref{adxprop:getcand3rdfdc:complexity}).
For the running time, constructing \(\*R'\) takes \(O(n (n+m))\) time (by Prop.~A.1), and generating \(\*R''\) takes \(O(n^3  (n+m))\) time (by Prop.~A.4).
% By Prop.~\ref{adxprop:cpg}, creating a causal path graph \(\G'\) relative to \((\G,\*X,\*Y)\) (by calling \textsc{GetCausalPathGraph}) takes \(O(n+m)\) time.
By Prop.~A.5, creating a causal path graph \(\G'\) relative to \((\G,\*X,\*Y)\) (by calling \textsc{GetCausalPathGraph}) takes \(O(n+m)\) time.
\textsc{TestSep} takes \(O(n+m)\) time.
The dominant factor is \(O(n^3 (n+m))\).
\end{proof}

% ListFDSets - Soundness & Completeness
% \begin{adxtheorem}{\ref{thm:listfdsets}}[Correctness of \textsc{ListFDSets}]
\begin{adxtheorem}[Correctness of \textsc{ListFDSets}]
Let \(\G\) be a causal graph, \(\*X,\*Y\) disjoint sets of variables, and \(\*I,\*R\) sets of variables.
\textsc{ListFDSets\((\G,\*X,\*Y, \*I, \*R)\)} enumerates all and only sets \(\*Z\) with \(\*I \subseteq \*Z \subseteq \*R\) that satisfy the FD criterion relative to \((\*X, \*Y)\) in \(O(n^4 (n+m))\) delay where \(n\) and \(m\) represent the number of nodes and edges in \(\G\).
\end{adxtheorem}

\begin{proof}
Consider the recursion tree for \textsc{ListFDSets}.
By induction on tree nodes, we show that when a tree node \(\mathcal{N}(\*I',\*R')\) is visited, \textsc{ListFDSets} will output all and only FD adjustment sets \(\*Z\) relative to \((\*X,\*Y)\) where \(\*I' \subseteq \*Z \subseteq \*R'\).

\begin{itemize}
    \item \textit{Base case}: Consider any leaf tree node \(\mathcal{L}(\*I',\*R')\).
    The recursion stops when \(\*I = \*R\), so \(\*I' = \*R'\) must hold.
    \(\mathcal{L}\) contains a node \(\*Z\) with \(\*Z = \*I' = \*R'\) if \(\*Z\) is a valid FD adjustment set relative to \((\*X,\*Y)\), or empty otherwise.
    % Indeed, \textsc{ListFDSets} will output a FD adjustment set \(\*Z\) if and only if \textsc{FindFDSet} in line~\ref{alg:listfdsets:findfdset} does not output \(\perp\) (by Thm.~\ref{adxthm:findfdset}).
    Indeed, \textsc{ListFDSets} will output a FD adjustment set \(\*Z\) if and only if \textsc{FindFDSet} in line~3 does not output \(\perp\) (by Thm.~1).
    
    \item \textit{Inductive case}: Let \(\mathcal{N}(\*I',\*R')\) be any non-leaf tree node.
    Assume the claim holds for two children of \(\mathcal{N}\).
    We show that \(\mathcal{N}\) contains all FD adjustment sets \(\*Z\) with \(\*I' \subseteq \*Z \subseteq \*R'\), which can also be expressed as the union of two collections of sets: 1) the collection of FD adjustment sets \(\*Z_1\) with \(\*I' \cup \{v\} \subseteq \*Z_1 \subseteq \*R'\), and 2) the collection of FD adjustment sets \(\*Z_2\) with \(\*I' \subseteq \*Z_2 \subseteq \*R' \setminus \{v\}\).
    The two collections are disjoint as every set in the first collection contains \(v\), and none in the second collection does.
    By assumption, each child contains the collection of respective FD adjustment sets.
    % If \textsc{FindFDSet} in line~\ref{alg:listfdsets:findfdset} outputs \(\perp\), then there does not exist a FD adjustment set \(\*Z\) with \(\*I' \subseteq \*Z \subseteq \*R'\).
    If \textsc{FindFDSet} in line~3 outputs \(\perp\), then there does not exist a FD adjustment set \(\*Z\) with \(\*I' \subseteq \*Z \subseteq \*R'\).
    Otherwise, each child outputs a respective collection of FD adjustment sets.
\end{itemize}

For the runtime, consider the recursion tree for \textsc{ListFDSets}.
% Every time a tree node \(\mathcal{N}(\*I',\*R')\) is visited, the function \textsc{FindFDSet} is called, which takes \(O(n^3 (n+m))\) time (by Thm.~\ref{adxthm:findfdset}).
Every time a tree node \(\mathcal{N}(\*I',\*R')\) is visited, the function \textsc{FindFDSet} is called, which takes \(O(n^3 (n+m))\) time (by Thm.~1).
If \textsc{FindFDSet} outputs \(\perp\), then \textsc{ListFDSets} does not search further from  \(\mathcal{N}\) because there exists no FD adjustment set \(\*Z\) with  \(\*I' \subseteq \*Z \subseteq \*R'\).
Otherwise, recursion continues until a leaf tree node is visited.
In each level of the tree, a single node \(v\) is removed from the set \(\*R \setminus \*I\).
The depth of the tree is at most \(n\), and the time required to output a set \(\*Z\) is \(O(n^4 (n+m))\).
In the worst case scenario, \(n\) branches will be aborted (i.e., \textsc{FindFDSet} outputs \(\perp\) on every level of the tree) before reaching the first leaf.
It takes \(O(n^4 (n+m))\) time to produce either the first output or halt.
Thus, \textsc{ListFDSets} runs with \(O(n^4 (n+m))\) delay.
\end{proof}